\documentclass[11pt]{article}
\pdfoutput=1
\usepackage{latexsym}
\usepackage{amsmath}
\usepackage{amssymb}
\usepackage{amsfonts,amsthm}
\usepackage[top=1in, bottom=1in, left=1in, right=1in]{geometry}
\usepackage{tikz}
\usetikzlibrary{arrows,decorations.pathmorphing,decorations.shapes,backgrounds,fit}
\pgfkeys{tikz/cc1/.style={dotted} }
\pgfkeys{tikz/cc4/.style={dashed} }
\pgfkeys{tikz/cc3/.style={decorate,decoration=bumps} }
\pgfkeys{tikz/cc2/.style={decorate,decoration=snake} }
\pgfkeys{tikz/solid/.style={line width=2pt} }
\pgfkeys{tikz/solid2/.style={line width=1pt} }
\pgfkeys{tikz/nodostile1/.style={draw,line width=1pt}}
\pgfkeys{tikz/nodostile2/.style={draw,line width=2pt}}
\pgfkeys{tikz/arcostile1/.style={very thick} }
\pgfkeys{tikz/arcostile2/.style={red,very thick} }
\pgfkeys{tikz/arcostile3/.style={thick} }
\pgfkeys{tikz/arcostile4/.style={line width=2pt} }

\newcommand{\ignore}[1]{}

\newtheorem{theorem}{Theorem}[section]
\newtheorem{lemma}[theorem]{Lemma}

\newtheorem{corollary}[theorem]{Corollary}

\usepackage{array}
\newcolumntype{C}[1]{>{\centering\let\newline\\\arraybackslash\hspace{0pt}}m{#1}}

\begin{document}
\title{\bf 2-Vertex Connectivity in Directed Graphs}
\author{Loukas Georgiadis$^{1}$ \and Giuseppe F. Italiano$^{2}$ \and Luigi Laura$^{3}$ \and Nikos Parotsidis$^{1}$}
\maketitle
\thispagestyle{empty}

\begin{abstract}
Given a directed graph, two vertices $v$ and $w$ are $2$-\emph{vertex-connected} if there are two internally vertex-disjoint paths from $v$ to $w$ and two internally vertex-disjoint paths from $w$ to $v$.
In this paper, we show how to compute this relation in $O(m+n)$ time, where $n$ is the number of vertices and $m$ is the number of edges of the graph. As a side result, we show how to build in linear time an $O(n)$-space data structure, which can answer in constant time queries on whether any two vertices are $2$-vertex-connected.
Additionally,  when two query vertices $v$ and $w$ are not $2$-vertex-connected, our data structure can produce in constant time a ``witness'' of this property, by exhibiting
a vertex or an edge that is contained in all paths from $v$ to $w$ or in all paths from $w$ to $v$.
We are also able to compute in linear time a sparse certificate for $2$-vertex connectivity, i.e., a subgraph of the input graph that has $O(n)$ edges and maintains the same $2$-vertex connectivity properties
as the input graph.
\end{abstract}

\footnotetext[1]{Department of Computer Science \& Engineering, University of Ioannina, Greece. E-mail: \texttt{\{loukas,nparotsi\}@cs.uoi.gr}.}
\footnotetext[2]{Dipartimento di Ingegneria Civile e Ingegneria Informatica, Universit\`a di Roma ``Tor Vergata'', Roma, Italy. E-mail: \texttt{giuseppe.italiano@uniroma2.it}.
Partially supported by MIUR, the Italian Ministry of Education, University and Research, under Project AMANDA
(Algorithmics for MAssive and Networked DAta).}
\footnotetext[3]{Dipartimento di Ingegneria Informatica, Automatica e Gestionale, ``Sapienza'' Universit\`a di Roma, Roma, Italy. E-mail: \texttt{laura@dis.uniroma1.it}.}

\section{Introduction}
\label{sec:introduction}

Let $G=(V,E)$ be a directed graph (digraph), with $m$ edges and $n$ vertices. $G$ is \emph{strongly connected} if there is a directed path from each vertex to every other vertex.
The \emph{strongly connected components} of $G$ are its maximal strongly connected subgraphs. Two vertices $u,v \in V$  are \emph{strongly connected}  if they belong to the same strongly connected component of $G$.
A vertex (resp., an edge) of $G$ is a \emph{strong articulation point} (resp., a \emph{strong bridge}) if its removal increases the number of strongly connected components.
A digraph $G$ is $2$-vertex-connected if it has at least three vertices and no strong articulation points; $G$ is $2$-edge-connected if it has no strong bridges. The $2$-vertex- (resp., $2$-edge-) connected components of $G$ are its maximal $2$-vertex- (resp., $2$-edge-) connected subgraphs.

\begin{figure}
\begin{center}
\includegraphics[width=\textwidth]{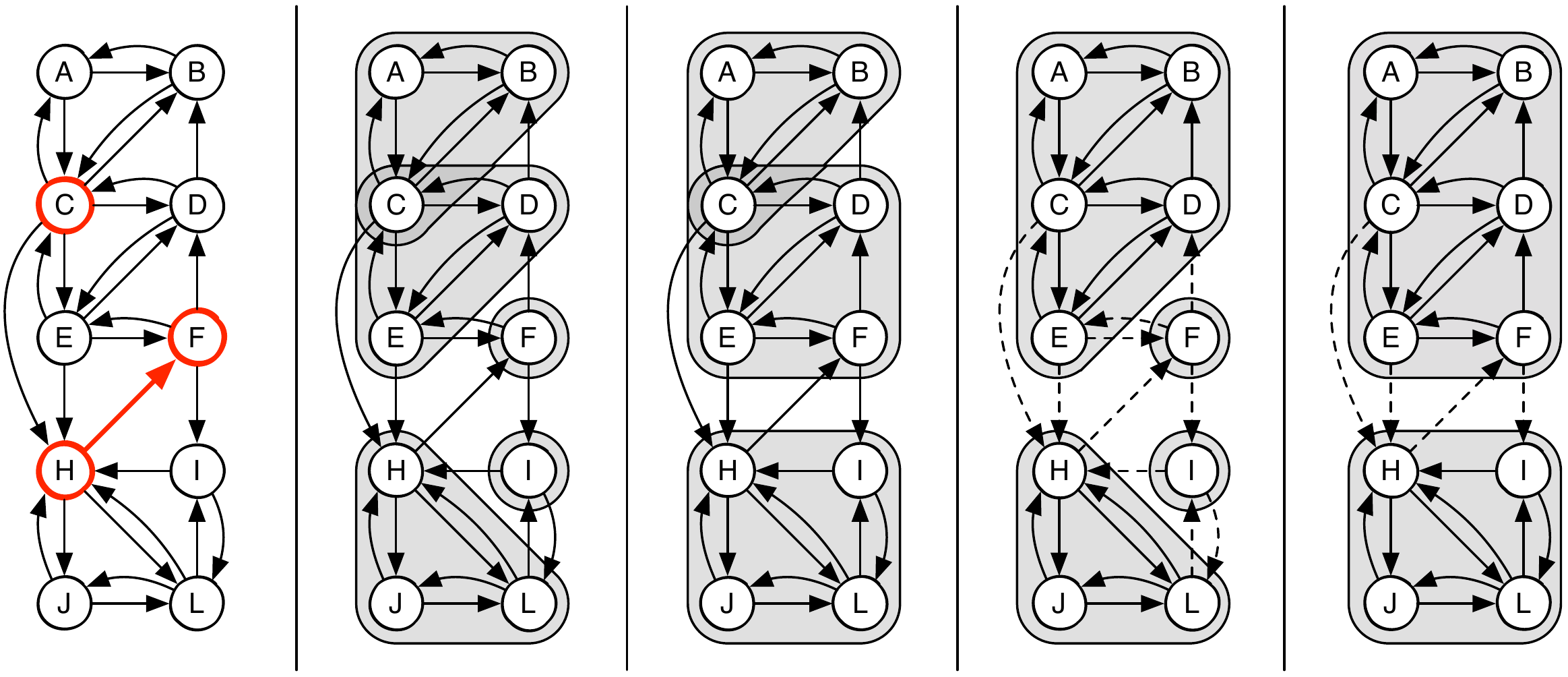}
\begin{tabular}{C{0.0\textwidth}C{0.10\textwidth}C{0.25\textwidth}C{0.15\textwidth}C{0.20\textwidth}C{0.15\textwidth}C{0.20\textwidth}}
\\[-.6cm]
& (a) $G$ &  (b) $2\mathit{VCC}(G)$  & (c) $2\mathit{VCB}(G)$ &  (d) $2\mathit{ECC}(G)$ &  (e) $2\mathit{ECB}(G)$ &
\end{tabular}
\end{center}
\vspace{-0.6cm}
\caption{(a) A strongly connected digraph $G$, with strong articulation points and strong bridges shown in red (better viewed in color). (b) The $2$-vertex-connected components of $G$. (c) The $2$-vertex-connected blocks of $G$. (d) The $2$-edge-connected components of $G$. (e) The $2$-edge-connected blocks of $G$.
}
\label{fig:example}
\end{figure}

Differently from undirected graphs, in digraphs $2$-vertex and $2$-edge connectivity have a much richer and more complicated structure.
To see an example of this,
let $v$ and $w$ be two distinct vertices and
consider the following natural $2$-vertex and $2$-edge connectivity relations, defined in~\cite{2ECB,2vcb:jaberi14,strong-k-blocks:rs81}.
Vertices $v$ and $w$ are said to be \emph{$2$-vertex-connected} (resp., \emph{$2$-edge-connected}), and we denote this relation by  $v \leftrightarrow_{\mathrm{2v}} w$ (resp., $v \leftrightarrow_{\mathrm{2e}} w$), if there are two internally vertex-disjoint  (resp., two edge-disjoint) directed paths from $v$ to $w$  and two internally vertex-disjoint  (resp., two edge-disjoint) directed paths from $w$ to $v$ (note that a path from $v$ to $w$ and a path from $w$ to $v$ need not be edge- or vertex-disjoint).
A \emph{$2$-vertex-connected block} (resp., \emph{$2$-edge-connected block}) of a digraph $G=(V,E)$ is defined as a maximal subset $B \subseteq V$ such that $u \leftrightarrow_{\mathrm{2v}} v$ (resp., $u \leftrightarrow_{\mathrm{2e}} v$) for all $u, v \in B$.
In undirected graphs, the  $2$-vertex- (resp., $2$-edge-) connected blocks are identical to the $2$-vertex- (resp., $2$-edge-) connected components.
As shown in Figure~\ref{fig:example}, this is not the case for digraphs.
Put in other words, differently from the undirected case, in digraphs $2$-vertex- (resp., $2$-edge-) connected components do not encompass the notion of pairwise $2$-vertex (resp., $2$-edge) connectivity among its vertices. We note that pairwise $2$-connectivity is relevant in several applications, where one is often interested in local properties, e.g., checking whether two vertices are $2$-connected, rather than in global properties.

It is thus not surprising that
$2$-connectivity problems on directed graphs
appear to be more difficult than on undirected graphs.
For undirected graphs
it has been known for over 40 years how to compute all bridges, articulation points, $2$-edge- and $2$-vertex-connected components in linear time,
by simply using depth first search~\cite{dfs:t}.
In the case of digraphs, however, the very same problems have been  much more challenging.
Indeed,
it has been shown only few years ago that all strong bridges and strong articulation points of a digraph can be computed in linear time~\cite{Italiano2012}. Furthermore,
the best current bound for computing the $2$-edge- and the $2$-vertex-connected components in digraphs is not even linear, but it is $O(n^2)$, and it was achieved only very recently by Henzinger et al.~\cite{2CC:HenzingerKL14},
improving previous $O(mn)$ time bounds~\cite{2vcc:jaberi14,ni93}.
Finally, it was shown also very recently how to compute the $2$-edge-connected blocks of digraphs in linear time \cite{2ECB}.

In this paper, we complete the picture on $2$-connectivity for digraphs by presenting the 
first algorithm for computing
the $2$-vertex-connected
blocks in $O(m+n)$ time. Our bound is
asymptotically optimal and it improves sharply over a previous $O(mn)$ time bound by  Jaberi \cite{2vcb:jaberi14}.
As a side result, our algorithm constructs an $O(n)$-space data structure
that reports in constant time if two vertices are $2$-vertex-connected.
Additionally,  when two query vertices $v$ and $w$ are not $2$-vertex-connected, our data structure can produce, in constant time, a ``witness'' 
by exhibiting
a vertex (i.e., a strong articulation point) or an edge (i.e., a strong bridge) that separates them.
We are also able to compute in linear time a sparse certificate for $2$-vertex connectivity, i.e., a subgraph of the input graph that has $O(n)$ edges and maintains the same $2$-vertex connectivity properties.
%
%
Our algorithm follows the high-level approach of \cite{2ECB} for computing the $2$-edge-connected blocks.
%
However, the algorithm for computing the $2$-vertex-connected blocks is much more involved and requires several novel  ideas and non-trivial techniques to achieve the claimed bounds.
In particular, the main technical difficulties  that need to be tackled when following the approach of \cite{2ECB} are the following:
\begin{itemize}
\item First, the algorithm in \cite{2ECB} maintains a partition of the vertices into approximate blocks, and refines this partition as the algorithm progresses.
Unlike $2$-edge-connected blocks, however,
$2$-vertex-connected blocks do not partition the vertices of a digraph, and therefore it is harder to maintain approximate blocks throughout the algorithm's execution. To cope with this problem, we show that these blocks can be maintained using a more complicated forest representation, and we define a set of suitable operations on this representation in order to refine and split blocks. We believe that our forest representation of the $2$-vertex-connected blocks of a digraph can be of independent interest.

\item Second, in \cite{2ECB} we used
a properly defined \emph{canonical decomposition} of the input digraph $G$, in order to obtain smaller \emph{auxiliary} digraphs (not
necessarily subgraphs of $G$) that maintain the original $2$-edge-connected blocks of $G$.
A key property of this decomposition was the fact that any vertex in an auxiliary graph $G_r$ is reachable from a vertex outside $G_r$ only through a single strong bridge. In the computation of the
$2$-vertex-connected blocks, we have to decompose the graph according to strong articulation points, and so the above crucial property is completely lost. To overcome this problematic issue, we need to design and to implement efficiently a different and more sophisticated decomposition.

\item Third, differently from $2$-edge connectivity, $2$-vertex connectivity in digraphs is plagued with several degenerate special cases, which are not only more tedious but also more cumbersome to deal with. For instance, the algorithm in \cite{2ECB} exploits implicitly the property that
two vertices $v$ and $w$ are $2$-edge-connected if and only if the removal of any edge leaves $v$ and $w$ in the same strongly connected component. Unfortunately, this property no longer holds for $2$-vertex connectivity, as for instance two mutually adjacent vertices are always left in the same strongly connected component by the removal of any other vertex, but they are not necessarily $2$-vertex-connected. To handle this more complicated situation, we introduce the notion of \emph{vertex-resilient blocks} and
prove some useful properties about the vertex-resilient and $2$-vertex-connected blocks of a digraph.
\end{itemize}

Another difference with \cite{2ECB} is that now we are able to provide a witness for two vertices not being $2$-vertex-connected. This approach can be applied to provide a witness for two vertices not being $2$-edge-connected, thus extending the result in \cite{2ECB}.
As in \cite{2ECB}, some basic components of our algorithms are flow graphs and dominator trees, that we review in Section \ref{sec:dominators}. In Section \ref{sec:blocks} we prove some useful properties of the vertex-resilient and $2$-vertex-connected blocks that allow us to represent them by a forest. Our linear-time algorithms for computing the vertex-resilient blocks and the $2$-vertex-connected blocks are described in Sections \ref{section:vertex-resilient-blocks} and \ref{section:2vc-blocks}. We describe the computation of the sparse certificate in Section \ref{section:sparse-certificate}.

\section{Flow graphs, dominators, and bridges}
\label{sec:dominators}

A \emph{flow graph} is a digraph such that every vertex is reachable from a distinguished start vertex. Let $G=(V,E)$ be the input digraph, which we assume to be strongly connected. (If not, we simply treat each strongly connected component separately.) For any vertex $s \in V$, we denote by $G(s)=(V,E,s)$ the corresponding flow graph with start vertex $s$; all vertices in $V$ are reachable from $s$ since $G$ is strongly connected. The \emph{dominator relation} in $G(s)$
is defined as follows: A vertex $u$ is a \emph{dominator} of a vertex $w$ ($u$ \emph{dominates} $w$) if every path from $s$ to $w$ contains $u$; $u$ is a \emph{proper dominator} of $w$ if $u$ dominates $w$ and $u \not= w$.  The dominator relation is reflexive and transitive. Its transitive reduction is a rooted tree, the \emph{dominator tree} $D(s)$: $u$ dominates $w$ if and only if $u$ is an ancestor of $w$ in $D(s)$. If $w \not= s$, $d(w)$, the parent of $w$ in $D(s)$, is the \emph{immediate dominator} of $w$: it is the unique proper dominator of $w$ that is dominated by all proper dominators of $w$. An edge $(u,w)$ is a \emph{bridge} in $G(s)$ if all paths from $s$ to $w$ include $(u,w)$.

Lengauer and Tarjan~\cite{domin:lt} presented an algorithm for computing dominators in  $O(m \alpha(n, m/n))$ time for a flow graph with $n$ vertices and $m$ edges, where $\alpha$ is a functional inverse of Ackermann's function~\cite{dsu:tarjan}.
Subsequently, several linear-time algorithms
were discovered~\cite{domin:ahlt,dominators:bgkrtw,domin:bkrw,dominators:Fraczak2013,minset:poset,dom:gt04}.
Italiano et al. \cite{Italiano2012} showed that the strong articulation points of $G$ can be computed from the dominator trees of $G(s)$ and $G^R(s)$, where $s$ is an arbitrary start vertex and $G^R$ is the digraph that results from $G$ after reversing edge directions; similarly, the strong bridges of $G$ correspond to the bridges of $G(s)$ and $G^R(s)$.

Let $T$ be a rooted tree whose vertex set is $V$. Tree $T$ has the \emph{parent property} if for all $(v, w) \in E$, $v$ is a descendant of the parent of $w$ in $T$.
Tree $T$ has the \emph{sibling property} if $v$ does not dominate $w$ for all siblings $v$ and $w$. The parent and sibling properties are necessary and sufficient for a tree to be the dominator tree~\cite{domcert}.

\ignore{
\begin{theorem}
\label{theorem:parent-sibling} \emph{(\cite{domcert})}
Tree $D$ has the parent and sibling properties.
\end{theorem}
}

\section{Vertex-resilient blocks and $2$-vertex-connected blocks}
\label{sec:blocks}

Let $v$ and $w$ be two distinct vertices in a digraph.
By Menger's Theorem~\cite{menger}, $v \leftrightarrow_{\mathrm{2e}} w$ if and only if
the removal of any edge leaves $v$ and $w$ in the same strongly connected component, i.e., two vertices are $2$-edge-connected if and only if they are resilient to the deletion of a single edge. The situation for $2$-vertex connectivity is more complicated.
Indeed, Menger's Theorem implies that $v \leftrightarrow_{\mathrm{2v}} w$ only if the removal of any vertex different from $v$ and $w$ leaves them in the same strongly connected component, while the converse holds only when $v$ and $w$ are not adjacent. For instance,
two mutually adjacent vertices are left in the same strongly connected component by the removal of any other vertex, although they are not necessarily $2$-vertex-connected. To handle this situation,
we use the following notation, which was also considered in~\cite{2vcb:jaberi14}.
Vertices $v$ and $w$ are said to be \emph{vertex-resilient},
denoted by $v \leftrightarrow_{\mathrm{vr}} w$
if the removal of any vertex different from $v$ and $w$
leaves $v$ and $w$ in the same strongly connected component. We define a \emph{vertex-resilient block}
of a digraph $G=(V,E)$ as a maximal subset $B \subseteq V$ such that $u \leftrightarrow_{\mathrm{vr}} v$
for all $u, v \in B$. See Figure \ref{fig:vrb-example}.
%
Note that, as a (degenerate) special case, a vertex-resilient block might consist of a singleton vertex only: we denote this as a \emph{trivial vertex-resilient block}. In the following, we will consider only non-trivial vertex-resilient blocks. Since there is no danger of ambiguity, we will call them simply vertex-resilient blocks. We remark that
two vertices $v$ and $w$ that are vertex-resilient are not necessarily $2$-vertex-connected: this is indeed the case for vertices $H$ and $F$ in the digraph of Figure \ref{fig:example}(a).  If, however, $v$ and $w$ are not adjacent then $v \leftrightarrow_{\mathrm{2v}} w$ if and only if  $v \leftrightarrow_{\mathrm{vr}} w$.

\begin{figure}
\begin{center}
\includegraphics[width=0.22\textwidth]{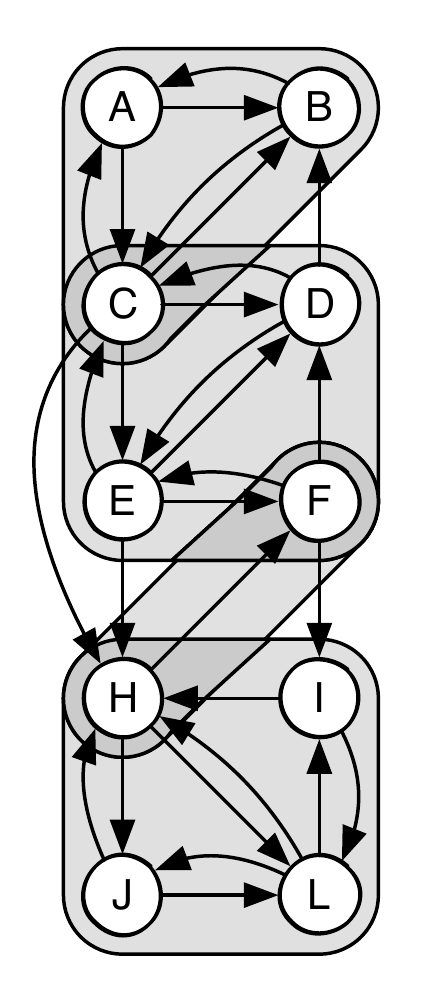}
\end{center}
\vspace{-0.6cm}
\caption{The vertex-resilient blocks of the digraph of Figure \ref{fig:example}.}
\label{fig:vrb-example}
\end{figure}

We next provide some basic properties of the vertex-resilient blocks and the $2$-vertex-connected blocks. In particular, we show that any digraph has at most $n-1$ vertex-resilient (resp., $2$-vertex-connected) blocks and, moreover, that there is a forest representation of these blocks that enables us to test vertex-resilience (resp., $2$-vertex-connectivity) between any two vertices in constant time. This structure is reminiscent of the representation used in \cite{onlineBiconnected:WT92} for the biconnected components of an undirected graph.

\begin{lemma}
\label{lemma:VR}
Let $u$, $v$, $x$, and $y$ be distinct vertices such that $u{\leftrightarrow_{\mathrm{vr}}}x$ , $v{\leftrightarrow_{\mathrm{vr}}}x$, $u{\leftrightarrow_{\mathrm{vr}}}y$ and $v{\leftrightarrow_{\mathrm{vr}}}y$. Then also $x {\leftrightarrow_{\mathrm{vr}}} y$ and $u {\leftrightarrow_{\mathrm{vr}}} v$.
\end{lemma}
\begin{proof}
Assume, for contradiction, that $x$ and $y$ are not vertex-resilient. Then there is a strong articulation point $w$ such that every path from $y$ to $x$ contains $w$, or every path from $x$ to $y$ contains $w$ (or both). Without loss of generality, suppose that $w$ is contained in every path from $y$ to $x$.
Since $u$ and $v$ are distinct, we can assume that $w \not= u$. (If $w = u$ then we swap the role of $u$ and $v$.) Then, $y {\leftrightarrow_{\mathrm{vr}}} u$ implies that there is a path $P$ from $y$ to $u$ that avoids $w$, and similarly, $u {\leftrightarrow_{\mathrm{vr}}} x$ implies that there is a path $Q$ from $u$ to $x$ that avoids $w$.
So, $P$ followed by $Q$ gives a path from $y$ to $x$ that does not contain $w$, a contradiction. Hence $x {\leftrightarrow_{\mathrm{vr}}} y$. The fact that $u {\leftrightarrow_{\mathrm{vr}}} v$ follows by repeating the same argument for $u$ and $v$.
\end{proof}

\begin{corollary}
\label{corollary:VRB}
Let $B$ and $B'$ be two distinct vertex-resilient blocks of a digraph $G=(V,E)$. Then $|B \cap B'| \le 1$.
\end{corollary}
\begin{proof}
Follows immediately from Lemma \ref{lemma:VR}.
\end{proof}

We denote by $\mathit{VRB}(u)$ the vertex-resilient blocks that contain $u$.
Define the \emph{block graph} $F=(V_F, E_F)$ of $G$ as follows. The vertex set $V_F$ consists of the vertices in $V$ and also contains one \emph{block node} for each vertex-resilient block of $G$. The edge set $E_F$ consists of the edges $\{u, B\}$ where $B \in \mathit{VRB}(u)$. Thus, $F$ is an undirected bipartite graph. Next we show that it is also acyclic.

\begin{lemma}
\label{lemma:VB-graph-path}
Let $u$ and $v$ be any vertices that are connected by a path $P$ in $F$. Then, for any vertex $w \in V$ not on $P$, $u$ and $v$ are strongly connected in digraph $G \setminus w$.
\end{lemma}
\begin{proof}
It suffices to show that $G$ contains a path $Q$ from $u$ to $v$ that avoids $w$. The same argument shows that $G$ contains a path from $v$ to $u$ that avoids $w$.
Let $P=(u_1=u, B_1, u_2, B_2, \ldots, u_{k+1}=v)$. Then $u_i \leftrightarrow_{\mathrm{vr}} u_{i+1}$, for $1 \le i \le k$, so there is a path $P_i$ in $G$ from $v_i$ to $v_{i+1}$ that avoids $w$.
Then the catenation of paths $P_1, \ldots, P_{k}$ gives a path in $G$ from $u$ to $v$ that avoids $w$.
\end{proof}

\begin{lemma}
\label{lemma:VB-graph}
Graph $F$ is acyclic.
\end{lemma}
\begin{proof}
Suppose, for contradiction, that $F$ contains a cycle $C$. We show that all vertices $w \in C \cap V$ belong to the same vertex-resilient block $B$.
Let $u, v \in V$ be two vertices on a minimal cycle $C$ of $F$ that are adjacent to a block node $B$. (Such $u$, $v$, and $B$ exist since $F$ is bipartite.) Then, $u$ and $v$ cannot be the only vertices in $V$ that are on $C$, since otherwise they would be adjacent to another block $B'$ on $C$, violating Corollary \ref{corollary:VRB}. Therefore, $C$ contains a vertex $w \in V \setminus \{u, v\}$. Clearly, $w \notin B$, otherwise the edge $\{ w, B\}$ would exist contradicting the minimality of $C$.
Hence, there is a vertex $z \in B$ such that all paths from $z$ to $w$ contain a common strong articulation point or all paths from $w$ to $z$ contain a common strong articulation point. Suppose, without loss of generality, that a vertex $x$ is contained in every path from $z$ to $w$.  Let $P$ be the path that results from $C$ by removing $B$. Let
$P_u$ and $P_v$ be the subpaths of $P$ from $u$ to $w$ and from $v$ to $w$, respectively. Then $x \not\in P_u$ or  $x \not\in P_v$ (or both). Suppose $x \not\in P_u$; if not then swap the role of $u$ and $v$.
Then, by Lemma \ref{lemma:VB-graph-path} there is a path $Q$ in $G$ from $u$ to $w$ that avoids $x$.
Also, since $u \leftrightarrow_{\mathrm{vr}} z$, there is a path $Q'$ in $G$ from $z$ to $u$ that avoids $x$.
Then the catenation of $Q'$ and $Q$ gives a path in $G$ from $z$ to $w$ that avoids $x$, a contradiction.
\end{proof}

\begin{lemma}
\label{lemma:blocks-number}
The number of vertex-resilient blocks in a digraph $G$ is at most $n-1$ .
\end{lemma}
\begin{proof}
We prove the lemma by showing that forest $F$ contains at most $n-1$ block nodes.
Since $F$ is a forest we can root each tree $T$ of $F$ at some arbitrary vertex $r$. Every level of $T$ contains either only vertices of $V$ or only block nodes, because $F$ is bipartite. Moreover, every block node is adjacent to at least two vertices of $V$, due to the fact that each (non-trivial) vertex-resilient block in $G$ contains at least $2$ vertices. Hence, every leaf of $T$ is a vertex in $V$.
Now consider a partition of $T$ into vertex disjoint paths $P_1, P_2, \ldots, P_k$, such that each $P_i$ leads from some vertex or block node to a leaf descendant. The number of block nodes in each $P_i$ is at most equal to $|P_i \cap V|$. Also, in the path $P_i$ starting at $r$ the number of block nodes in $P_i$ is less than $|P_i \cap V|$.  We conclude that there at most $n-1$ block nodes in $F$.
\end{proof}

\begin{lemma}
\label{lemma:blocks-size}
The total number of vertices in all vertex-resilient blocks is at most $2n-2$.
\end{lemma}
\begin{proof}
By Lemmas \ref{lemma:VB-graph} and \ref{lemma:blocks-number}, the block graph $F$ is a forest with at most $2n-1$ vertices. Each occurrence of a vertex $v$ in a block $B$ corresponds to an edge $\{v,B\}$ of $F$. Therefore, the total number of vertices in all vertex-resilient blocks equals the number of edges in $F$, and the lemma follows.
\end{proof}

\begin{figure}[t!]
\begin{center}
\includegraphics[width=0.43\textwidth]{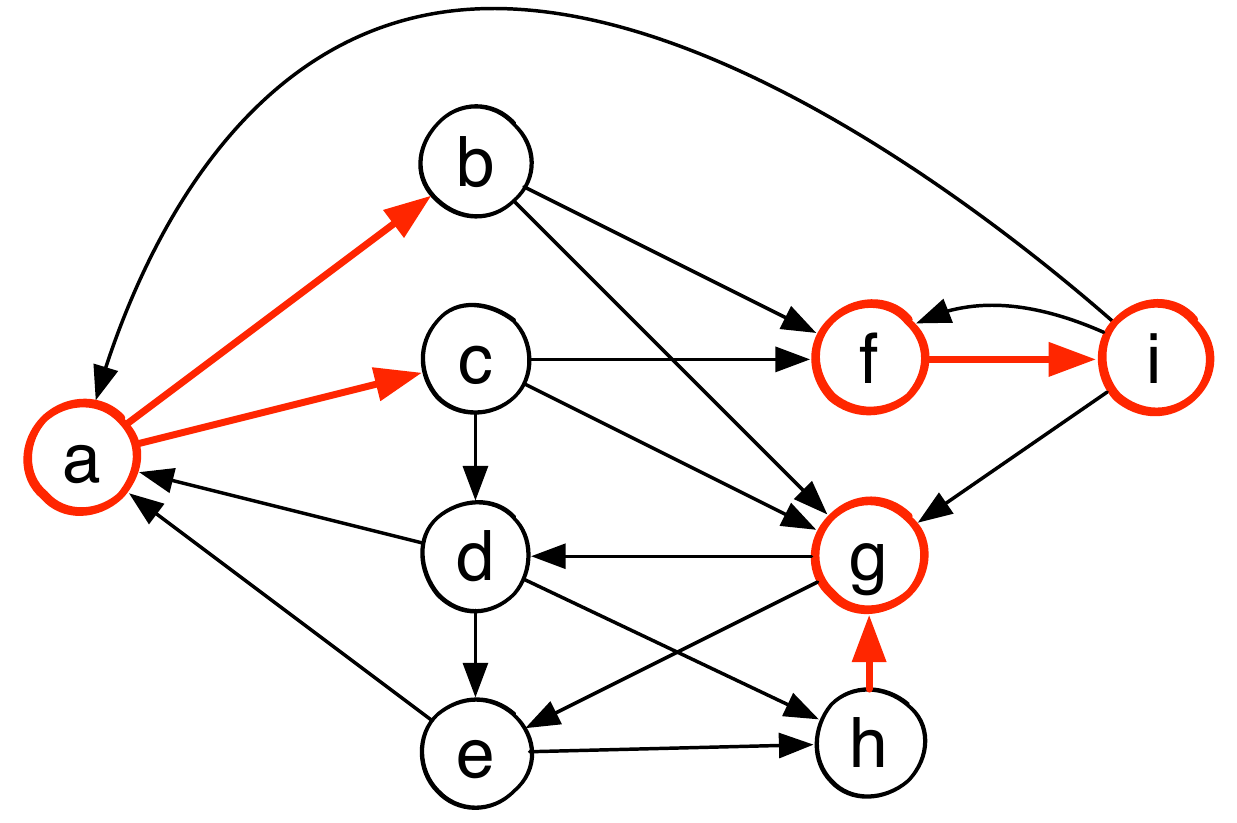}
\includegraphics[width=0.55\textwidth]{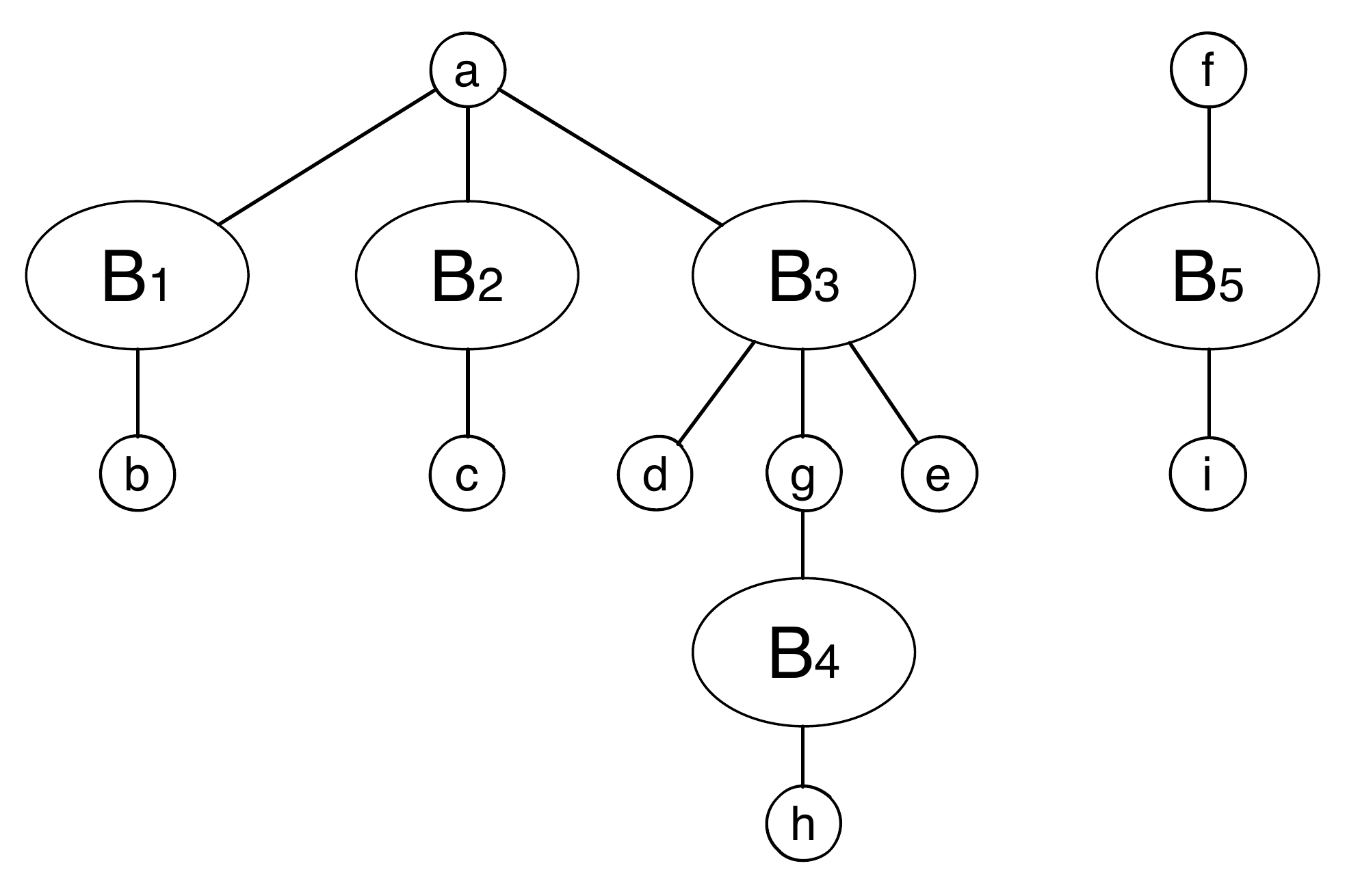}
\caption{A digraph $G$ and its vertex-resilient block forest $F$. The strong articulation points and the strong bridges of $G$ are shown in red. (Better viewed in color.)\label{fig:vrb-forest}}
\end{center}
\end{figure}

\begin{lemma}
\label{lemma:VB-graph-path-2}
Let $u$ and $v$ be any vertices that are not vertex-resilient but are connected by a path $P$ in $F$. Then, for any vertex $w \in V \setminus \{ u, v \}$ on $P$, $u$ and $v$ are not strongly connected in digraph $G \setminus w$.
\end{lemma}
\begin{proof}
We prove the lemma by contradiction. Let $P$ be a path that connects $u$ and $v$ in $F$. By Lemma \ref{lemma:VB-graph}, this path is unique for $u$ and $v$. First suppose that $P$ contains only one other vertex $w \in V \setminus \{u,v\}$, so $P=(u,B,w,B',v)$. Then $u \leftrightarrow_{\mathrm{vr}} w$ and $w \leftrightarrow_{\mathrm{vr}} v$. Now suppose that $u$ and $v$ are strongly connected in $G \setminus w$. This fact, together with Lemma \ref{lemma:VB-graph-path}, imply that $u$ and $v$ are strongly connected in $G \setminus x$ for all $x \in V \setminus \{u,v\}$. But this contradicts the assumption that $u$ and $v$ are not vertex-resilient.

Now suppose that path $P$ contains more than one vertex in $V \setminus \{u,v\}$. Let $P=(u=w_0, B_1, w_1, \ldots, B_{k}, w_{k}, B_{k+1}, v=w_{k+1})$, where $k>1$. By the argument above, $w_{i-1}$ and $w_{i+1}$ are not strongly connected in $G \setminus w_{i}$ for all $i \in \{1,\ldots,k\}$. Suppose that $u$ and $v$ are strongly connected in $G \setminus w_i$ for a fixed $i \in \{1,\ldots,k\}$. By Lemma \ref{lemma:VB-graph-path}, $u$ and $w_{i-1}$, and $w_{i+1}$ and $v$, are strongly connected in $G \setminus w_i$. But then, $w_{i-1}$ and $w_{i+1}$ are also strongly connected in $G \setminus w_i$, a contradiction.
\ignore{
Let $x$ be any vertex on $Q$ other than $u$ and $v$. Then $u \leftrightarrow_{\mathrm{vr}} w$ implies that $G$ contains a path $Q_1$ from $u$ to $w$ that avoids $x$, and $w \leftrightarrow_{\mathrm{vr}} v$ implies that $G$ contains a path $Q_2$ from $w$ to $v$ that avoids $x$. Hence $Q_1 \cdot Q_2$ is a path from $u$ to $v$ that avoids $x$. The same argument shows that for any vertex $x \in Q' \setminus \{u,v\}$, there is a path from $v$ to $u$ that avoids $x$.
We argue that if $G \setminus w$ contains a path from $u$ to $v$ and a path from
It suffices to show that $G$ contains a path $Q$ from $u$ to $v$ that avoids $w$. The same argument shows that $G$ contains a path from $v$ to $u$ that avoids $w$.
Let $P=(u_1=u, B_1, u_2, B_2, \ldots, u_{k+1}=v)$. Then $u_i \leftrightarrow_{\mathrm{vr}} u_{i+1}$, for $1 \le i \le k$, so there is a path $P_i$ in $G$ from $v_i$ to $v_{i+1}$ that avoids $w$.
Then the catenation of paths $P_1, \ldots, P_{k}$ gives a path in $G$ from $u$ to $v$ that avoids $w$.}
\end{proof}

We consider $F$ as a forest of rooted trees by choosing an arbitrary vertex as the root of each tree.
Then $u \leftrightarrow_{\mathrm{vr}} w$ if and only if $u$ and $w$ are siblings or one the grandparent of the other. See Figure \ref{fig:vrb-forest}.
We can perform both tests in constant time simply by storing the parent of each vertex in $F$.
Thus, we can test in constant time if two vertices are vertex-resilient. Note that we cannot always apply Lemma \ref{lemma:VB-graph-path-2} to find a strong articulation point that separates two vertices $u$ and $w$ that are not vertex-resilient. Indeed, two vertices that are strongly connected but not vertex-resilient may not even be connected by a path in the forest $F$ (see, e.g., vertices $f$ and $h$ in Figure \ref{fig:vrb-forest}). So if we wish to return a witness that $u$ and $w$ are not vertex-resilient, we cannot rely on $F$. We deal with this problem in Section \ref{sec:queries}.

Now we turn to $2$-vertex-connected blocks and provide some properties that enable us to compute them via the vertex-resilient blocks.

\begin{lemma}
\label{lemma:2vc-resilient-sb}
Let $v$ and $w$ be two distinct vertices of $G$ such that $v \leftrightarrow_{\mathrm{vr}} w$. Then, $v$ and $w$ are not $2$-vertex connected if and only if at least one of the edges $(v,w)$ and $(w,v)$ is a strong bridge in $G$.
\end{lemma}
\begin{proof}
Menger's Theorem~\cite{menger} implies that if $v$ and $w$ are not adjacent then $v \leftrightarrow_{\mathrm{2v}} w$ if and only if  $v \leftrightarrow_{\mathrm{vr}} w$. If, on the other hand, $v \leftrightarrow_{\mathrm{vr}} w$ but $v$ and $w$ are not $2$-vertex-connected, then at least one of the edges $(v,w)$ and $(w,v)$ exists in $G$ and is a strong bridge.
\end{proof}

The following corollary, which relates $2$-vertex-connected, $2$-edge-connected and vertex-resilient blocks, is an immediate consequence of Lemma \ref{lemma:2vc-resilient-sb}.

\begin{corollary}
\label{cor:2vc-resilient}
For any two distinct vertices $v$ and $w$, $v \leftrightarrow_{\mathrm{2v}} w$ if and only if $v \leftrightarrow_{\mathrm{vr}} w$ and $v \leftrightarrow_{\mathrm{2e}} w$.
\end{corollary}

\ignore{
\begin{lemma}
\label{lemma:2vc-resilient}
For any two distinct vertices $v$ and $w$, $v \leftrightarrow_{\mathrm{2v}} w$ if and only if $v \leftrightarrow_{\mathrm{vr}} w$ and $v \leftrightarrow_{\mathrm{2e}} w$.
\end{lemma}
}

By Corollary \ref{cor:2vc-resilient} we have that the $2$-vertex-connected blocks are refinements of the vertex-resilient blocks, formed by the intersections of the vertex-resilient blocks and the $2$-edge-connected blocks of the digraph $G$. Since the $2$-edge-connected blocks are a partition of the vertices of $G$, these intersections partition each vertex-resilient block.
From this property we conclude that Lemmas \ref{lemma:VR},  \ref{lemma:VB-graph}, and \ref{lemma:blocks-number} and Corollary \ref{corollary:VRB} also hold for the $2$-vertex-connected blocks.

\section{Computing the vertex-resilient blocks}
\label{section:vertex-resilient-blocks}

In this section we present new algorithms for computing the vertex-resilient blocks of a digraph $G$. We can assume that $G$ is strongly connected, so $m \ge n$. If not, then we process each strongly connected component separately; if $u \leftrightarrow_{\mathrm{vr}} v$ then $u$ and $v$ are in the same strongly connected component $S$ of $G$, and moreover, any vertex on a path from $u$ to $v$ or from $v$ to $u$ also belongs in $S$.
We begin with a simple algorithm that removes a single strong articulation point at a time. In order to get a more efficient solution, we need to consider simultaneously how different strong articulation points divide the vertices into blocks, which we do with the help of dominator trees. We achieve linear running time by combining the simple algorithm with the dominator-tree-based division, and by applying suitable operations on the block forest structure.

\subsection{A simple algorithm}
\label{section:vr-simple}

Algorithm \textsf{SimpleVRB}, illustrated in Figure \ref{fig:SimpleVRB}, is an immediate application of the characterization of the vertex-resilient blocks in terms of strong articulation points.
Let $u$ and $v$ be two distinct vertices. We say that a strong articulation point $x$ \emph{separates $u$ from $v$} if all paths from $u$ to $v$ contain $x$. In this case $u$ and $v$ belong to different strongly connected components of $G\setminus x$.
This observation implies that we can compute the vertex-resilient blocks by computing the strongly connected components of $G\setminus x$ for every strong articulation point $x$.
To do this efficiently we define an operation that refines the currently computed blocks.
Let $\mathcal{B}$ be a set of blocks, let $\mathcal{S}$ be a partition of a set $U \subseteq V$, and let $x$ be a vertex not in $U$.

\begin{list}{}{
\setlength{\leftmargin}{1.7cm} \setlength{\labelsep}{.2cm} \setlength{\itemsep}{0cm}\setlength{\labelwidth}{1.5cm}
}
\item[$\mathit{refine}(\mathcal{B}, \mathcal{S}, x)$:] For each block $B \in \mathcal{B}$, substitute $B$ by the sets  $B \cap (S \cup \{ x \} )$ of size at least two, for all  $S \in \mathcal{S}$.
\end{list}

\noindent
In Section \ref{section:2vc-blocks}, where we will  compute the $2$-vertex-connected blocks from the vertex-resilient blocks and the $2$-edge-connected blocks,
we will use the notation $\mathit{refine}(\mathcal{B}, \mathcal{S})$ as a shorthand for $\mathit{refine}(\mathcal{B}, \mathcal{S}, x)$ with $x = \mathit{null}$.

\begin{lemma}
\label{lemma:refine}
Let $N$ be the total number of elements in all sets of $\mathcal{B}$ ($N=\sum_{B \in \mathcal{B}}|B|$), and let $K$ be the number of elements in $U$. Then, the operation $\mathit{refine}(\mathcal{B}, \mathcal{S}, x)$ can be executed in $O(N+K)$ time.
\end{lemma}
\begin{proof}
A simple way to achieve the claimed bound is to number the sets of the partition $\mathcal{S}$, each with a distinct integer id in the interval $[1,K]$.
Consider a block $B$. Each element $v \in B$ is assigned a label that is equal to the id of the set $S \in \mathcal{S}$ that contains  $v$ if $v \in U$, and zero otherwise.
Then, the computation of the sets $B \cap (S \cup \{ x \} )$ for all $S \in \mathcal{S}$ can be done in $O(|B|)$ time with bucket sorting.
\end{proof}

\begin{figure}[t!]
\begin{center}
\fbox{
\begin{minipage}[h]{\textwidth}
\begin{center}
\textbf{Algorithm \textsf{SimpleVRB}: Computation of the vertex-resilient blocks of a strongly connected digraph $G=(V,E)$}
\end{center}
\begin{description}\setlength{\leftmargin}{10pt}
\item[Step 1:] Compute the strong articulation points of $G$.
\item[Step 2:] Initialize the current set of blocks as $\mathcal{B} = \{ V \}$. (Start from the trivial set containing only one block.)
\item[Step 3:] For each strong articulation point $x$ do:
  \begin{description}\setlength{\leftmargin}{10pt}
    \item[Step 3.1:] Compute the strongly connected components $S_1,\ldots,S_k$ of $G\setminus x$. Let $\mathcal{S}$ be the partition of $V\setminus x$ defined by the strongly connected components $S_i$.
    \item[Step 3.2:] Execute $\mathit{refine}(\mathcal{B}, \mathcal{S}, x)$.
  \end{description}
\end{description}
\end{minipage}
}
\caption{Algorithm \textsf{SimpleVRB}\label{fig:SimpleVRB}}
\end{center}
\end{figure}

\begin{lemma}
\label{lemma:SimpleVRB}
Algorithm \textsf{SimpleVRB} runs in $O(m p^{\ast})$ time, where $p^{\ast}$ is the number of strong articulation points of $G$. This is $O(m n)$ in the worst case.
\end{lemma}
\begin{proof}
The strong articulation points of $G$ can be computed in linear time by \cite{Italiano2012}. In each iteration of Step 3, we can compute the strongly connected components of $G\setminus x$ in linear time \cite{dfs:t}.
As we discover the $i$-th strongly connected component, we assign label $i$ ($i \in \{1,\ldots,n\}$) to the vertices in $S_i$.
By Lemma \ref{lemma:blocks-number}, the number of vertex-resilient blocks of $G$ is at most $n-1$. 
Therefore, since the total number of blocks (trivial and non-trivial) cannot decrease during any iteration, $\mathcal{B}$ contains at most $n-1$ blocks in each execution of Step 3. 
By induction on the number of iterations, it follows that the algorithm maintains the invariant that any two distinct blocks in $\mathcal{B}$ have at most one element in common, and that the corresponding block graph is a forest. Therefore, by Lemma \ref{lemma:blocks-number}, the total number of elements in all blocks is at most $2n-2$. So, by Lemma \ref{lemma:refine}, each iteration of Step 3 takes $O(n)$ time. This yields the desired $O(m p^{\ast})$ running time, where $p^*$ is the number of strong articulation points of $G$.
Since a digraph may have up to $n$ strong articulation points, this is $O(m n)$ in the worst case.
\end{proof}


\subsection{Linear-time algorithm}
\label{section:VRB-linear}

We will show how to obtain a faster algorithm by applying the framework developed in \cite{2ECB} for the computation of the $2$-edge-connected blocks, namely by using dominator trees and auxiliary graphs. As already mentioned, auxiliary graphs need to be defined in a substantially different way, which complicates several technical details. 

As a warm up, first consider the computation of $\mathit{VRB}(v)$, i.e., the vertex-resilient blocks that contain a specific vertex $v$. Consider the flow graph $G(v)$ with start vertex $v$ and its reverse digraph $G^R(v)$, obtained after reversing edge directions. Let $w$ be a vertex other than $v$. Clearly, $v$ and $w$ are vertex-resilient if and only if $v$ is the only proper dominator of $w$ in both $G(v)$ and $G^R(v)$, i.e., $d(w)=v$ and $d^R(w)=v$. Now let $u$ be a sibling of $w$ in both $D(v)$ and $D^R(v)$. The fact that $d^R(w)=v$ and $d(u)=v$ implies that for any vertex $x \in V \setminus \{v,w,u\}$ there is path from $w$ to $u$ through $v$ that avoids $x$. So $w$ and $u$ are in a common vertex-resilient block that contains $v$ if and only if they lie in the same strongly connected component of $G \setminus v$. This observation implies the following linear-time algorithm to compute the vertex-resilient blocks that contain $v$. Compute the dominator trees $D(v)$ and $D^R(v)$ of $G(v)$ and $G^R(v)$ respectively. Let $C(v)$ (resp., $C^R(v)$) be the set of children of $v$ in $G(v)$ (resp., $G^R(v)$). Set $U = C(v) \cap C^R(v)$ and initialize the set of blocks $\mathcal{B} = \{ U \}$. Compute the strongly connected blocks $S_1, S_2, \ldots, S_k$ of $G \setminus v$. Let $\mathcal{S}$ be the set that contains the nonempty restrictions of the $S_i$ sets to $U$, i.e., $\mathcal{S}$ contains the nonempty sets $S_i \cap U$. Finally, execute $\mathit{refine}(\mathcal{B}, \mathcal{S}, v)$.

Note that all the vertex-resilient blocks can be computed in $O(mn)$ time by applying the above algorithm to all vertices $v$.
To avoid the repeated applications of this algorithm we develop a new concept of \emph{auxiliary graphs} for $2$-vertex connectivity. Before doing that, we state two properties regarding information that a dominator tree can provide about vertex-resilient blocks and paths.

\begin{lemma}
\label{lemma:vertex-resilient-necessary}
Let $G=(V,E)$ be a strongly connected graph, and let $s \in V$ be an arbitrary start vertex. Any two vertices $x$ and $y$ are vertex-resilient only if they are siblings in $D(s)$ or one is the immediate dominator of the other in $G(s)$.
\end{lemma}
\begin{proof}
Immediate.
\end{proof}

\begin{lemma}
\label{lemma:blocks-paths}
Let $r$ be a vertex, and let $v$ be any vertex that is not a descendant of $r$ in $D(s)$. Then there is a path from $v$ to $r$ that does not contain any proper descendants of $r$ in $D(s)$. Moreover, all simple paths from $v$ to any descendant of $r$ in $D(s)$ contain $r$.
\end{lemma}
\begin{proof}
Let $P$ be any path from $v$ to $r$. (Such a path exists since graph $G$ is strongly connected.) Let $u$ be the first vertex on $P$ such that $u$ is a descendant of $r$. Then either $u=r$ or $u$ is a proper descendant of $r$. In the first case the lemma holds. Suppose $u$ is a proper descendant of $r$. Since $v$ is not a descendant of $r$ in $D(s)$, there is a path $Q$ from $s$ to $v$ in $G$ that does not contain $r$. Then $Q$ followed by the part of $P$ from $v$ to $u$ is a path from $s$ to $u$ that avoids $r$, a contradiction.
\end{proof}

\subsubsection{Auxiliary graphs}
\label{section:auxiliary}

As in \cite{2ECB}, \emph{auxiliary graphs} are a key concept in our algorithm that provides a decomposition of the input digraph $G$ into smaller digraphs (not necessarily subgraphs of $G$) that maintain the original vertex-resilient blocks.
In \cite{2ECB} we used a \emph{canonical decomposition} of the input digraph, in order to obtain auxiliary graphs that maintain the $2$-edge-connected blocks. A key property of this decomposition was the fact that any vertex in an auxiliary graph $G_r$ is reachable from a vertex outside $G_r$ only though a single strong bridge. In the computation of the vertex-resilient blocks, however, we have to decompose the input digraph according to strong articulation points, and thus the above property is completely lost. To overcome this critical issue, we apply a different and more involved decomposition.

Let $s$ be an arbitrarily chosen start vertex in $G$. Recall that we denote by $G(s)$ the flow graph with start vertex $s$, by $G^R(s)$ the flow graph obtained from $G(s)$ after reversing edge directions, by $D(s)$ and $D^R(s)$ the
dominator trees of $G(s)$ and $G^R(s)$ respectively, and by $C(v)$ and $C^R(v)$ the set of children of $v$ in $D(s)$ and $D^R(s)$ respectively.

For each vertex $r$, let $C^k(r)$ denote the level $k$ descendants of $r$, i.e., $C^0(r)=\{r\}$, $C^1(r)=C(r)$, etc. For each vertex $r \not=s$ that is not a leaf in $D(s)$ we build the \emph{auxiliary graph $G_r = (V_r, E_r)$ of $r$} as follows. The vertex set of $G_r$ is $V_r =\cup_{k=0}^{3} C^k(r)$ and it is partitioned into a set of \emph{ordinary} vertices  $V_r^o = C^1(r) \cup C^2(r)$ and a set of \emph{auxiliary} vertices $V_r^a = C^0(r) \cup C^3(r)$.
The auxiliary graph $G_r$ results from $G$ by contracting the vertices in $V \setminus V_r$ as follows. All vertices that are not descendants of $r$ in $D(s)$ are contracted into $r$. For each vertex $w \in C^3(r)$, we contract all descendants of $w$ in $D(s)$ into $w$.  See Figure~\ref{fig:auxiliary}.
We use the same definition for the auxiliary graph $G_s$ of $s$, with the only difference that we let $s$ be an ordinary vertex. Also note that when we form $G_s$ from $G$, no vertex is contracted into $s$.
In order to bound the size of all auxiliary graphs, we eliminate parallel edges during those contractions.

\begin{lemma}
\label{lemma:auxiliary-graphs-size}
The auxiliary graphs $G_r$ have at most $4n$ vertices and $4m+n$ edges in total.
\end{lemma}
\begin{proof}
A vertex of $G$ may appear in at most four auxiliary graphs. Therefore, the total number of edges in all auxiliary graphs excluding type-(b) shortcut edges $(u,v)$ with $u \not \in V_r$ is at most $4m$. A type-(b) shortcut edge $(u,v)$ with $u \not \in V_r$ of $G_r$ corresponds to a unique vertex in $C^3(r)$, so there are at most $n$ such edges.
\end{proof}

\ignore{
For each vertex $r \not=s$ that is not a leaf in $D(s)$ we build the \emph{auxiliary graph $G_r = (V_r, E_r)$ of $r$} as follows. Let $C^k(r)$ denote the level $k$ descendants of $r$, i.e., $C^0(r)=\{r\}$, $C^1(r)=C(r)$, etc. The vertex set of $G_r$ is $V_r =\cup_{k=0}^{3} C^k(r)$ and it is partition into a set of \emph{ordinary} vertices  $V_r^o = C^1(r) \cup C^2(r)$ and a set of \emph{auxiliary} vertices $V_r^a = C^0(r) \cup C^3(r)$ .
The edge set $E_r$ contains all edges in $G=(V,E)$ induced by the vertices in $V_r$ (i.e., edges $(u,v)\in E$ such that $u \in V_r$ and $v \in V_r$). We also add in $E_r$ the following types of \emph{shortcut} edges that correspond to paths in $G$.
(a) If $G$ contains an edge $(u,v)$ such that $u \not\in V_r$ is a descendant of $r$ in $D(s)$ and $v \in V_r$ then we add the shortcut edge $(z,v)$ where $z$ the is an ancestor of $u$ in $D(s)$ such that $z \in C^3(r)$. (b) If $G$ contains an edge $(u,v)$ such that $u$ but not $v$ is a descendant of $r$ in $D(s)$ then we add the shortcut edge $(z,r)$ where $z$ the nearest ancestor of $u$ in $D(s)$ such that $z \in V_r$ ($z=u$ if $u \in V_r$). We note that we do not keep multiple (parallel) shortcut edges.  See Figure~\ref{fig:auxiliary}.
We use the same definition for the auxiliary graph $G_s$ of $s$, with the only difference that we let $s$ be an ordinary vertex. Also note that $G_s$ does not contain type-(b) shortcut edges.
}

\begin{figure}[t!]
\begin{center}
\begin{tabular}{cc}
\includegraphics[width=0.5\textwidth]{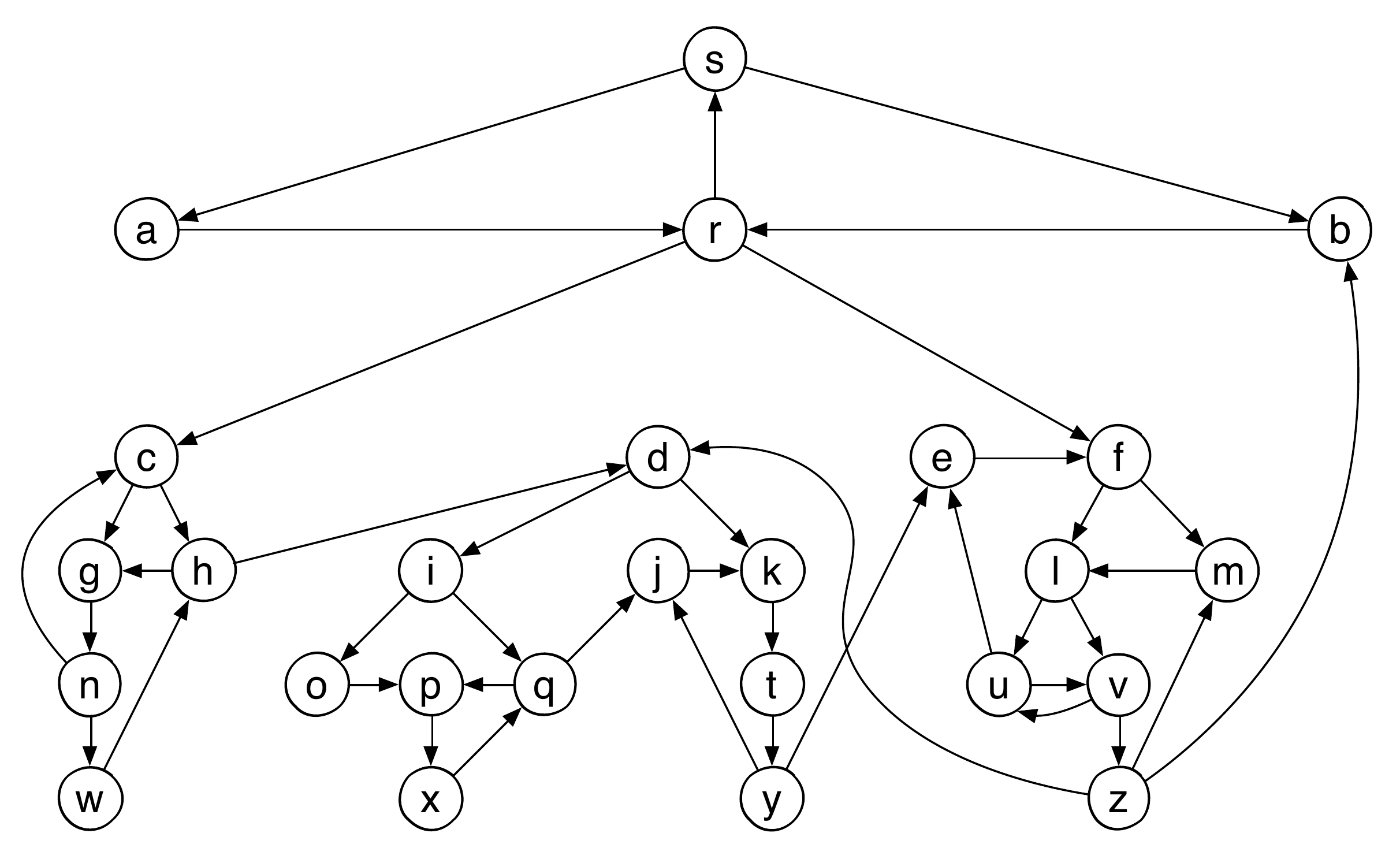} &
\includegraphics[width=0.5\textwidth]{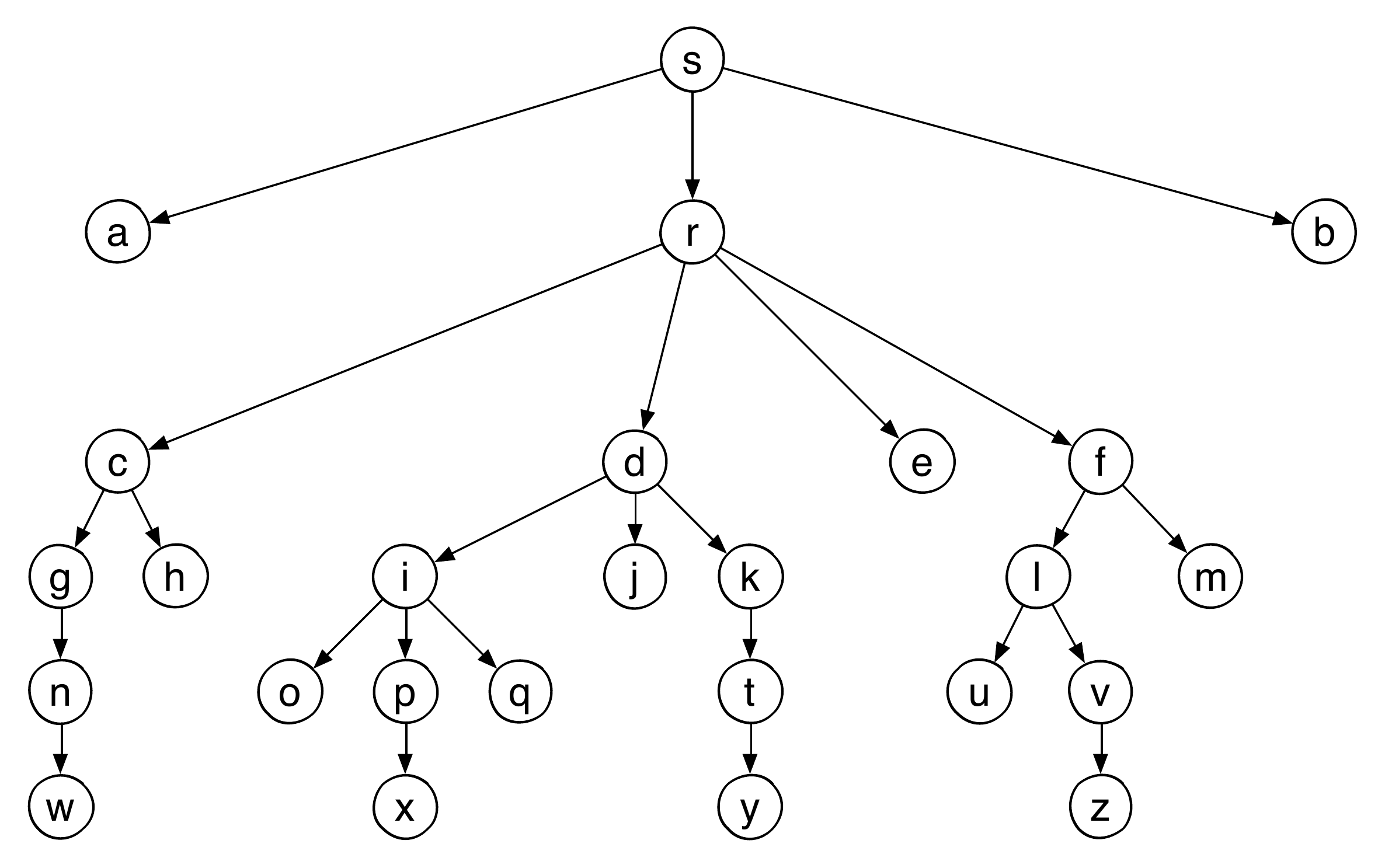}\\
$G$ & $D(s)$\\
\vspace{2ex}
\\
\includegraphics[width=0.5\textwidth]{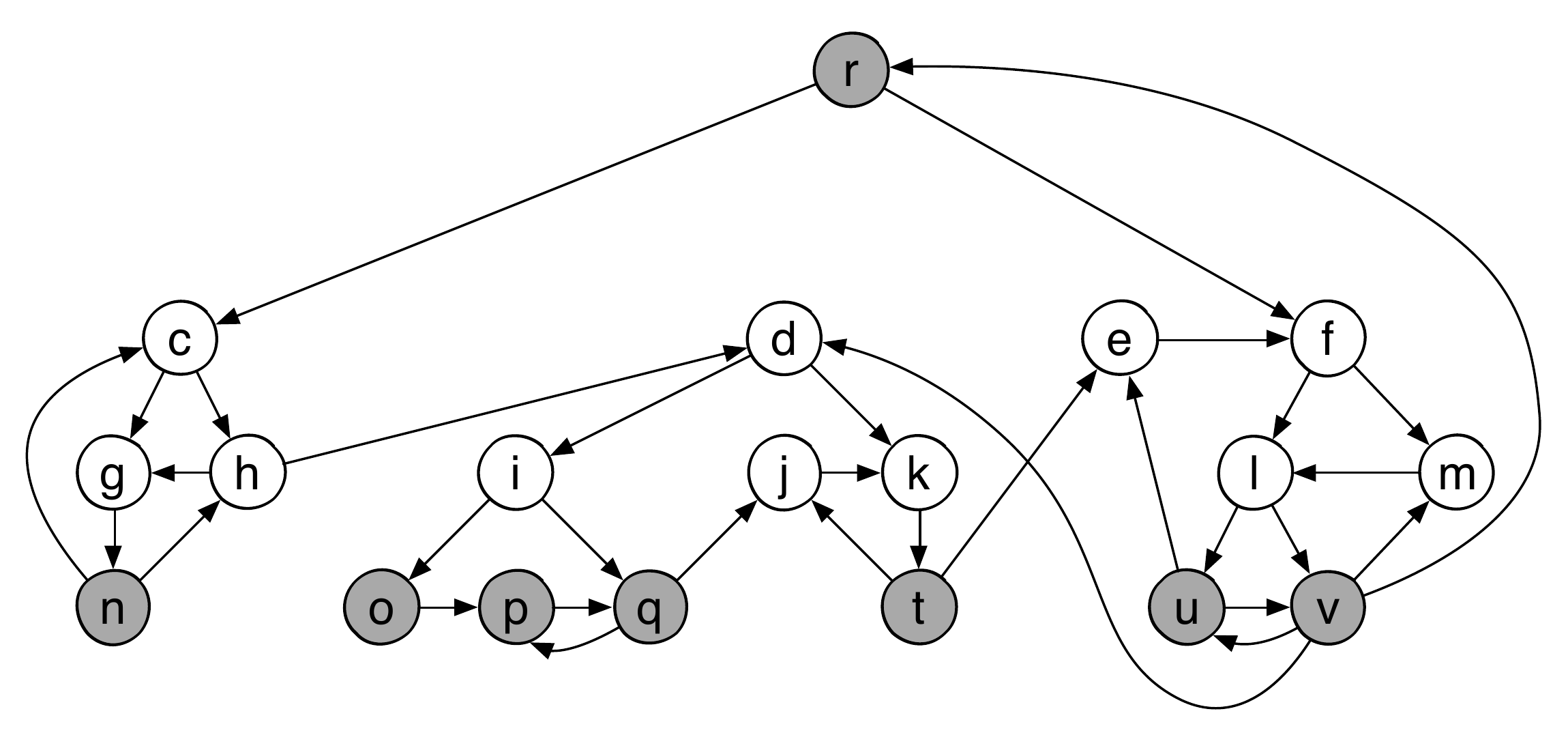} &
\includegraphics[width=0.23\textwidth]{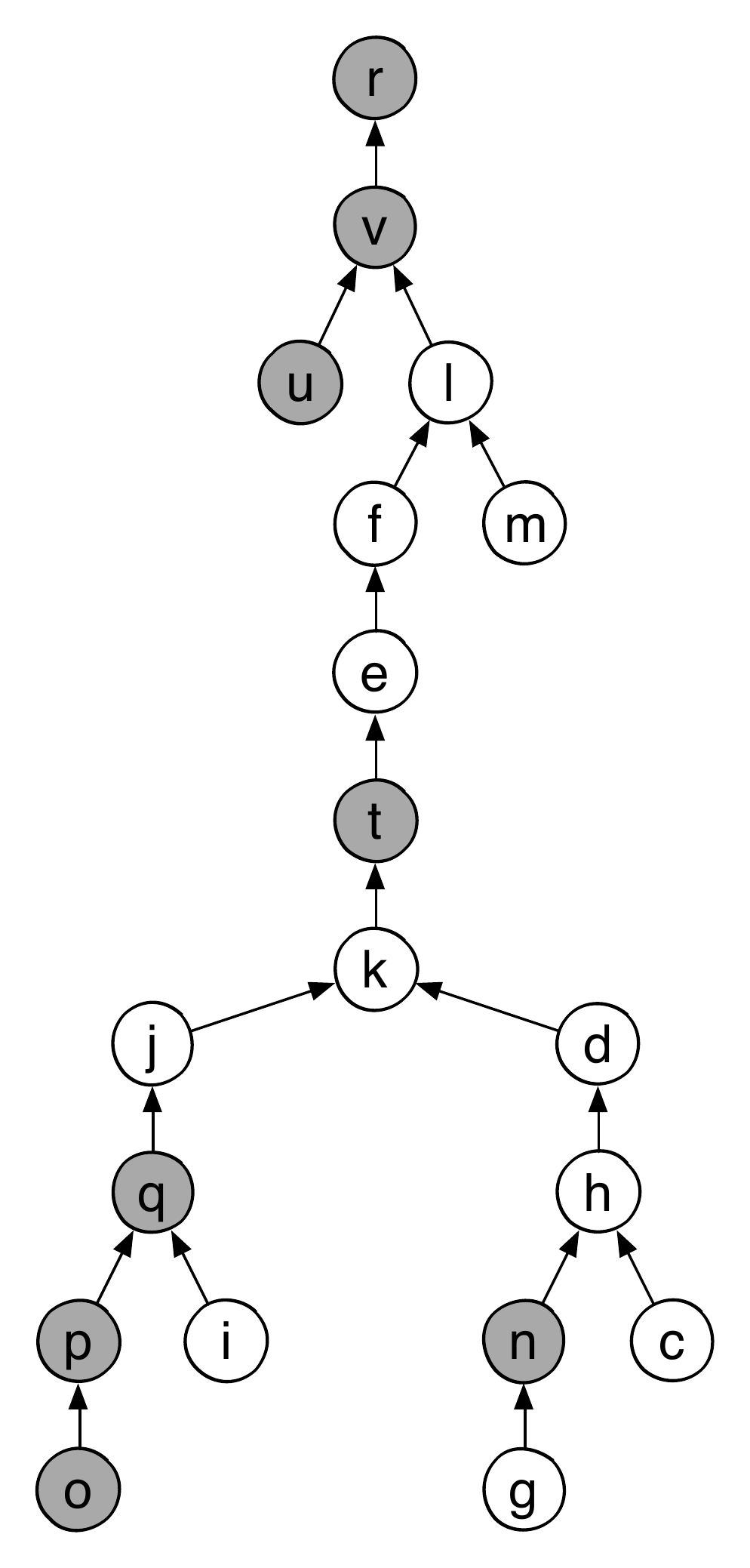}\\
$H=G_r$ & $D_H^R(r)$
\end{tabular}
\caption{A strongly connected graph $G$, the dominator tree $D(s)$ of flow graph $G(s)$, the auxiliary graph $H=G_r$ and the dominator tree $D_H^R(r)$ of the flow graph $H^R(r)$. (The edges of the dominator tree $D_H^R(r)$ are shown directed from child to parent.) The auxiliary vertices of $H$ are shown gray. \label{fig:auxiliary}}
\end{center}
\end{figure}

\begin{lemma}
\label{lemma:blocks-auxiliary-paths}
Let $v$ and $w$ be two vertices in $V_r$. Any path $P$ from $v$ to $w$ in $G$ has a corresponding path $P_r$ from $v$ to $w$ in $G_r$, and vice versa. Moreover, if $v$ and $w$ are both ordinary vertices in $G_r$, then $P_r$ contains a strong articulation point if and only if $P$ does.
\end{lemma}
\begin{proof}
The correspondence between paths in $G$ and paths in $G_r$ follows from the definition of the auxiliary graph. Next we prove the second part of the lemma. Let $P_r$ be the path in $G_r$ that corresponds to a path $P$ from $v$ to $w$ in $G$, where both $v$ and $w$ are ordinary vertices in $G_r$. By the construction of the auxiliary graph, we have that if $P_r$ contains a strong articulation point then so does $P$.
For the contraposition, suppose $P$ contains a strong articulation point $x$. Consider the following cases:
\begin{itemize}
\item $x \in V_r$.  Then, by the construction of the auxiliary graph, we have $x \in P_r$.
\item $x$ is a descendant of a vertex $z \in C^3(r)$. Vertex $z$ is a strong articulation point since it is either $x$ or a proper descendant of $x$.
Then, by Lemma  \ref{lemma:blocks-paths}, the part of $P$ from $v$ to $x$ contains $z$. So, $P_r$ also contains $z$ by the construction of the auxiliary graph.
\item $x$ is not a descendant of $r$. In this case, we have $r \not= s$. Since $v$ and $w$ are ordinary vertices of $G_r$, $C^1(r)$ is not empty and therefore $r$ is a strong articulation point.
By Lemma  \ref{lemma:blocks-paths}, the part of $P$ from $x$ to $w$ contains $r$. So, $P_r$ also contains $r$ by the construction of the auxiliary graph.
\end{itemize}
Hence, in every case $P_r$ contains a strong articulation point and the lemma follows.
\end{proof}

\begin{corollary}
\label{corollary:auxialiry-graphs}
Each auxiliary graph $G_r$ is strongly connected.
\end{corollary}
\begin{proof}
Follows from the construction of $G_r$, Lemma \ref{lemma:blocks-auxiliary-paths}, and the fact that $G$ is strongly connected.
\end{proof}

The next lemma shows that auxiliary graphs maintain the vertex-resilient relation of the original digraph.

\begin{lemma}
\label{lemma:vertex-resilient-auxiliary}
Let $v$ and $w$ be any two distinct vertices of $G$. Then $v$ and $w$ are vertex-resilient in $G$ if and only if they are both ordinary vertices in an auxiliary graph $G_r$ and they are vertex-resilient in $G_r$.
\end{lemma}
\begin{proof}
Suppose first that $v$ or $w$ is $s$. Without loss of generality assume $v=s$. Then by Lemma \ref{lemma:vertex-resilient-necessary} we have that $w \in C^1(r)$, so $v$ and $w$ are both ordinary vertices of $G_s$.
Now consider that $v, w \in V\setminus s$.
From Lemma \ref{lemma:vertex-resilient-necessary} we have that $v$ and $w$ belong in a set $C^1(r) \cup C^2(r)$ so they are both ordinary vertices of $G_r$.
Clearly if all paths from $v$ to $w$ in $G_r$ contain a common vertex (strong articulation point), then so do all paths from $v$ to $w$ in $G$ by Lemma \ref{lemma:blocks-auxiliary-paths}. Now we prove the converse. Suppose all paths from  $v$ to $w$ in $G$ contain a common vertex $u$. If $u \in V_r$ then also all paths from $v$ to $w$ in $G_r$ contain $u$ by the proof of Lemma \ref{lemma:blocks-auxiliary-paths}.
So suppose $u \not \in V_r$. Then $v$ is not an ancestor of $w$ in $D(s)$, since otherwise there would be a path from $v$ to $w$ that avoids $u$.

First consider that $u$ is a (proper) descendant of $r$ in $D(s)$. Since $v$ is not an ancestor of $w$ in $D(s)$, there is a vertex $x \in C^3(r)$ that is an ancestor of $u$. By Lemma \ref{lemma:blocks-paths}, all paths from $v$ to $u$ in $G$, and thus all paths from $v$ to $w$, contain $x$. By Lemma \ref{lemma:blocks-auxiliary-paths} this is also true for all paths from $v$ to $w$ in $G_r$.

Finally, if $u$ is not a descendant of $r$, Lemma \ref{lemma:blocks-paths} implies that all paths from $u$ to $w$ in $G$ contain vertex $r$. Hence, all paths from $v$ to $w$ in $G$ contain $r$, and so do all paths from $v$ to $w$ in $G_r$ by Lemma \ref{lemma:blocks-auxiliary-paths}.
\end{proof}

\ignore{
\begin{lemma}
\label{lemma:auxiliary-graphs-size}
The auxiliary graphs $G_r$ have at most $4n$ vertices and $4m+n$ edges in total.
\end{lemma}
\begin{proof}
A vertex of $G$ may appear in at most four auxiliary graphs. Therefore, the total number of edges in all auxiliary graphs excluding type-(b) shortcut edges $(u,v)$ with $u \not \in V_r$ is at most $4m$. A type-(b) shortcut edge $(u,v)$ with $u \not \in V_r$ of $G_r$ corresponds to a unique vertex in $C^3(r)$, so there are at most $n$ such edges.
\end{proof}
}

Now we specify how to compute all the auxiliary graphs $G_r=(V_r,E_r)$ in $O(m+n)$ time. Observe that the edge set $E_r$ contains all edges in $G=(V,E)$ induced by the vertices in $V_r$ (i.e., edges $(u,v)\in E$ such that $u \in V_r$ and $v \in V_r$). We also add in $E_r$ the following types of \emph{shortcut} edges that correspond to paths in $G$.
(a) If $G$ contains an edge $(u,v)$ such that $u \not\in V_r$ is a descendant of $r$ in $D(s)$ and $v \in V_r$ then we add the shortcut edge $(z,v)$ where $z$ the is an ancestor of $u$ in $D(s)$ such that $z \in C^3(r)$. (b) If $G$ contains an edge $(u,v)$ such that $u$ but not $v$ is a descendant of $r$ in $D(s)$ then we add the shortcut edge $(z,r)$ where $z$ the nearest ancestor of $u$ in $D(s)$ such that $z \in V_r$ ($z=u$ if $u \in V_r$). We note that we do not keep multiple (parallel) shortcut edges.  See Figure~\ref{fig:auxiliary}.
We use the same definition for the auxiliary graph $G_s$ of $s$, with the only difference that we let $s$ be an ordinary vertex. We also note that $G_s$ does not contain type-(b) shortcut edges.

To construct the auxiliary graphs $G_r = (V_r, E_r)$ we need to specify how to compute the shortcut edges of type (a) and (b). To do this efficiently we need to test ancestor-descendant relations in $D(s)$. There are several simple $O(1)$-time tests of this relation~\cite{domin:tarjan}. The most convenient one for us is to number the vertices of $D(s)$ from $1$ to $n$ in preorder, and to compute the number of descendants of each vertex. Then, vertex $v$ is a descendant of $r$ if and only if $\mathit{pre}(r) \le \mathit{pre}(v) < \mathit{pre}(r) + \mathit{size}(r)$, where, for any vertex $x$, $\mathit{pre}(x)$ and $\mathit{size}(x)$ are, respectively, the preorder number and the number of descendants of $x$ in $D(s)$.

Suppose $(u,v)$ is an edge of type (a). We need to find the ancestor $z$ of $u$ in $D(s)$ such that $z \in C^3(r)$. We process all such arcs of $G_r$ as follows. We create a list $B_r$ that contains the edges $(u,v)$ of type (a), and sort $B_r$ in increasing preorder of $u$. We create a second list $B'_r$ that contains the vertices in $C^3(r)$, and sort $B'_r$ in increasing preorder. Then, the shortcut edge of $(u,v)$ is $(z,v)$, where $z$ is the last vertex in the sorted list $B'_r$ such that $\mathit{pre}(z) \le \mathit{pre}(u)$. Thus the shortcut edges of type (a) can be computed in linear time by bucket sorting and merging.
Now we consider the edges of type (b). For each vertex $w \in C^3(r)$ we need to test if there is an edge $(u,v)$ in $G$ such that $u$ is a proper descendant of $w$ and $v$ is not a descendant of $r$ in $D(s)$. In this case, we add in $G_r$ the edge $(w,r)$. To do this test efficiently, we assign to each edge $(u,v)$ a tag $t(u,v)$ which we set equal to the preorder number of the nearest common ancestor of $u$ and $v$ in $D(s)$. We can do this easily by using the parent property and the $O(1)$-time test of the ancestor-descendant relation as follows: $t(u,v) = \mathit{pre}(u)$ if $u$ is an ancestor of $v$ in $D(s)$, $t(u,v) = \mathit{pre}(v)$ if $v$ is an ancestor of $u$ in $D(s)$, and $t(u,v)=\mathit{pre}(d(v))$ otherwise. At each vertex $w \not= s$ in $D(s)$ we store a label $\ell(w)$ which is the minimum tag of among the edges $(w,v)$. Using these labels we compute for each $w \not= s$ in $D(s)$ the values $\mathit{low}(w) = \min\{ \ell(v) \ | \ v \mbox{ is a descendant of } w \mbox{ in } D(s) \}$. These computations can be done in $O(m)$ time by processing the tree $D(s)$ in a bottom-up order. Now consider the auxiliary graph $G_r$. We process the vertices in $C^3(r)$. For each such vertex $w$ we add the shortcut edge $(w,r)$ if $\mathit{low}(w) < \mathit{pre}(r)$.

\begin{lemma}
\label{lemma:auxiliary-graphs-construction}
We can compute all auxiliary graphs $G_r$ in $O(m)$ time.
\end{lemma}

\subsection{Algorithm}
\label{sec:algorithm}

Our linear-time algorithm \textsf{FastVRB} is illustrated in Figure \ref{fig:FastVRB}. It uses two levels of auxiliary graphs and applies one iteration of Algorithm \textsf{SimpleVRB} for each auxiliary graph of the second level.
The algorithm uses different dominator trees, and applies Lemma \ref{lemma:vertex-resilient-necessary} in order to identify the vertex-resilient blocks. Since different dominator trees may define different blocks (which by Lemma \ref{lemma:vertex-resilient-necessary} are supersets of the vertex-resilient blocks), we will use an operation that we call $\mathit{split}$ to combine the different blocks.

We begin by computing the dominator tree $D(s)$ for an arbitrary start vertex $s$. For any vertex $v$, we let $\hat{C}(v)$ denote the set containing $v$ and the children of $v$ in $D(s)$, i.e., $\hat{C}(v)=C(v) \cup \{v\}$. Lemma \ref{lemma:vertex-resilient-necessary} gives an initial division of the vertices into blocks that are supersets of the vertex-resilient blocks. Specifically, the vertex-resilient blocks that contain $v$ are subsets of $\hat{C}(v)$ or $\hat{C}(d(v))$ (for $v \not= s$).

During the course of the algorithm, each vertex $v$ becomes associated with a set of blocks $\mathcal{B}(v)$ that contain $v$, which are subsets of $\hat{C}(v)$ and $\hat{C}(d(v))$ if $v \not= s$. The blocks are refined by applying the $\mathit{refine}$ operation of Section \ref{section:vr-simple} and operation $\mathit{split}$ that we define next, and at the end of the algorithm each set of blocks $\mathcal{B}(v)$ will be equal to $\mathit{VRB}(v)$.

Let $B$ be a block and $T$ be a tree with vertex set $V(T) \supseteq B$. For any vertex $v \in V(T)$, let $\hat{C}_T(v)$ be the set containing $v$ and the children of $v$ in $T$.
\begin{list}{}{
\setlength{\leftmargin}{1.7cm} \setlength{\labelsep}{.2cm} \setlength{\itemsep}{0cm}\setlength{\labelwidth}{1.5cm}
}
\item[$\mathit{split}(B, T)$:] Return the set that consists of the blocks $B \cap \hat{C}_T(v)$ of size at least two, for all $v \in V(T)$.
\end{list}

\begin{lemma}
\label{lemma:split}
Let $N$ be the number of vertices in $V(T)$. Then, the operation $\mathit{split}(B, T)$ can be executed in $O(N)$ time.
\end{lemma}
\begin{proof}
We number the vertices of $T$ in preorder. Let $\mathit{pre}(v)$ be the preorder number of $v \in V(T)$. Let $t(v)$ be the parent of $v \not= r$ in $T$, where $r$ is the root of $T$. We associate each vertex $v \not= r$ in $B$ with two labels $\ell_1(v) = \mathit{pre}(t(v))$ and $\ell_2(v) = \mathit{pre}(v)$, and create two corresponding pairs $\left<\ell_1(v), v\right>$ and $\left<\ell_2(v), v\right>$. Also, if $r \in B$, we associate $r$ with one label $\ell_2(r) = \mathit{pre}(r)$, and create a corresponding pair $\left<\ell_2(r),r\right>$.  Each block created by the $\mathit{split}$ operation consists of a set of at least two vertices $v \in B$ that are associated with a specific label. We can find these blocks by sorting the pairs $\left<\ell_j(v),v\right>$ by label, which can be done in $O(N)$ time with bucket sort.
\end{proof}

\begin{figure}[t!]
\begin{center}
\fbox{
\begin{minipage}[h]{\textwidth}
\begin{center}
\textbf{Algorithm \textsf{FastVRB}: Linear-time computation of the vertex-resilient blocks of a strongly connected digraph $G=(V,E)$}
\end{center}
\begin{description}\setlength{\leftmargin}{10pt}
\item[Step 1:] Choose an arbitrary vertex $s \in V$ as a start vertex. Compute the dominator tree $D(s)$. For any vertex $v$, let $\hat{C}(v)$ be the set containing $v$ and the children of $v$ in $D(s)$. For every vertex $v$ that is not a leaf in $D(s)$, associate block $\hat{C}(v)$ with every vertex $w \in \hat{C}(v)$.
\item[Step 2:] Compute the auxiliary graphs $G_r$ for all vertices $r$ that are not leaves in $D(s)$.
\item[Step 3:] Process the vertices of $D(s)$ in bottom-up order. For each auxiliary graph $H = G_r$ with $r$ not a leaf in $D(s)$ do:
    \begin{description}\setlength{\leftmargin}{10pt}
    \item[Step 3.1:] Compute the dominator tree $T=D_H^R(r)$. 
    \item[Step 3.2:] Compute the set $\mathcal{B}$ of blocks that contain vertices in $C(r)$.
    \item[Step 3.3:] For each block $B \in \mathcal{B}$ execute $\mathit{split}(B, T)$.
    \item[Step 3.4:] Compute the auxiliary graphs $H^R_q$ for all vertices $q$ that are not leaves in $T$.
    \item[Step 3.5:] For each auxiliary graph $H^R_q$ with $q$ not a leaf do:
        \begin{description}\setlength{\leftmargin}{10pt}
        \item[Step 3.5.1:] Compute the set $\mathcal{B}_q$ of blocks that contain at least two ordinary vertices in $H^R_q$.
        \item[Step 3.5.2:] Compute the set $\mathcal{S}$ of the strongly connected components of $H^R_q\setminus q$.
        \item[Step 3.5.3:] Refine the blocks in $\mathcal{B}_q$ by executing $\mathit{refine}(\mathcal{B}_q, \mathcal{S}, q)$.
        \end{description}
    \end{description}
\end{description}
\end{minipage}
}
\caption{Algorithm \textsf{FastVRB}\label{fig:FastVRB}}
\end{center}
\end{figure}

At a high level, the algorithm begins with a ``coarse'' block tree, induced by the $\hat{C}(v)$ sets of $D(s)$, which is then refined by the blocks defined from the dominator trees of the auxiliary graphs. The final vertex-resilient block forest is then computed by considering the strongly connected components of the second level auxiliary graphs, after removing their designated start vertex.
The algorithms needs to keep track of the blocks that contain a specific vertex, and, conversely, of the vertices that are contained in a specific block. To facilitate this search we explicitly store the adjacency lists of the current block forest $F$. Recall that $F$ is bipartite, so the adjacency list of a vertex $v$ stores the blocks that contain $v$, and the adjacency list of a block node $B$ stores the vertices in $B$. Initially $F$ contains one block for each set $\hat{C}(v)$, for all vertices $v$ that are not leaves in $D(s)$. These blocks are later refined by executing the $\mathit{split}$ and $\mathit{refine}$ operations, which maintain the invariant that $F$ is a forest, and that any two distinct blocks have at most two vertices in common. When we execute a $\mathit{split}$ or a $\mathit{refine}$ operation we can update the adjacency lists of $F$, while maintaining the bounds given in Lemmas \ref{lemma:refine} and \ref{lemma:split}.
Also, since during the execution of the algorithm the number of blocks can only increase, $F$ contains at most $n-1$ blocks at any given time. This fact implies that Lemma \ref{lemma:blocks-number} holds, so the total number of vertices and edges in $F$ is $O(n)$.

\begin{lemma}
\label{lemma:FastVRB-correctness}
Algorithm \textsf{FastVRB} is correct.
\end{lemma}
\begin{proof}
Let $u$ and $v$ be any vertices. If $u$ and $v$ are vertex-resilient in $G$, then by Lemma \ref{lemma:vertex-resilient-auxiliary} they are vertex-resilient in both auxiliary graphs of $G$ and $G_r$ that contain them as ordinary vertices. This
implies that the algorithm will correctly include them in the same block in Step 1 and will not separate them in Steps 3.3 and 3.5. So suppose that $u$ and $v$ are not vertex-resilient.
Then, without loss of generality, we can assume that all  paths from $u$ to $v$ contain a common strong articulation point. Thus, $d(v) \not= u$. We argue that all the blocks that contain $u$ and all the blocks that contain $v$ will be separated in some step of the algorithm.

First we observe that $u$ and $v$ can appear together in at most one of the blocks constructed in Step 1. Also, $u$ and $v$ can remain in at most one block after each $\mathit{split}$ operation ($u$ and $v$ can have at most one identical label $\ell_i(u) = \ell_j(v)$). So suppose that $u$ and $v$ are still contained in one common block just before the execution of Step 3.5. We will show that $u$ and $v$ will be separated after the $\mathit{refine}$ operation executed in Step 3.5.3. Since $u$ and $v$ were not separated by a $\mathit{split}$ operation, they are either siblings or one is the parent of the other in $D^R_H(r)$. Also, since $d(v) \not= u$ we have the following cases.

(a) $d(u) = v$. Then $u$ and $v$ are both ordinary vertices of the auxiliary graph $H=G_r$ with $r = d(v)$. Lemma \ref{lemma:vertex-resilient-auxiliary} implies that $G_r$ contains a strong articulation point $x$ that separates $u$ from $v$. We argue that $x$ is a proper ancestor of $u$ in $D^R_H(r)$. If not, then $H^R$ contains a path $P^R$ from $u$ to $r$ that avoids $x$. Since $d(v)=r$, $H$ contains a path $Q$ from $r$ to $v$ that avoids $x$. Thus $P \cdot Q$ is a path in $H$ from $u$ to $v$ that avoids $x$, a contradiction. Now we claim that $q = d^R_H(u)$ is also a strong articulation point that separates $u$ from $v$. Suppose the claim is false. Then $x \not= q$, so $x$ is a proper ancestor of $q$ in $D^R_H(r)$. Let $P$ be a path from $u$ to $v$ that avoids $q$. Then $x$ is on $P$ since $x$ separates $u$ from $v$. Let $P_x$ be the part of $P$ from $u$ to $x$. Also, since $x$ is a proper ancestor of $q$ in $D^R_H(r)$, $H^R$ has a path $Q^R$ from $r$ to $x$ that avoids $q$.  Then $P \cdot Q$ is a path in $H$ from $u$ to $r$ that avoids $q$, a contradiction.
The claim implies that $u$ and $v$ are located in different strongly connected components of $H^R_q\setminus q$, so they are contained in different blocks computed in Step 3.5.3.

(b) $d(v) = d(u) = r$. Then $u$ and $v$ are both ordinary vertices of the auxiliary graph $H=G_r$. Lemma \ref{lemma:vertex-resilient-auxiliary} implies that $G_r$ contains a strong articulation point $x$ that separates $u$ from $v$. By the same arguments as in case (a), it follows that $q = d^R_H(u)$ is a strong articulation point that separates $u$ from $v$.  So again $u$ and $v$ will be located in different blocks after Step 3.5.3.
\end{proof}

\ignore{
In order to analyze the running time of Algorithm \textsf{FastVRB} we need to specify how to compute the set of blocks $\mathcal{B}_q$ in Step 3.4.1. To facilitate this search we explicitly store the adjacency lists of the current block forest $F$. Recall that $F$ is bipartite, so the adjacency list of a vertex $v$ stores the blocks that contain $v$, and the adjacency list of a block node $B$ stores the vertices in $B$. Initially $F$ contains one block for each set $\hat{C}(v)$, for all $v \in D(s)$. These blocks are later refined by executing the $\mathit{split}$ and $\mathit{refine}$ operations, which maintain the invariant that $F$ is a forest, and that any two distinct blocks have at most two vertices in common. Also, since during the execution of the algorithm the number of blocks can only increase, $F$ contains at most $n-1$ blocks at any given time. This fact implies that Lemma \ref{lemma:blocks-number} holds, so the total number of vertices and edges in $F$ is $O(n)$.
}

\begin{lemma}
\label{lemma:FastVRB-time}
Algorithm \textsf{FastVRB} runs in $O(m)$ time.
\end{lemma}
\begin{proof}
We account for the total time spent on each step that Algorithm \textsf{FastVRB} executes.
Step 1 takes $O(m)$ time by \cite{dominators:bgkrtw}, and Step 2 takes $O(m)$ time by Lemma \ref{lemma:auxiliary-graphs-construction}. From Lemma \ref{lemma:auxiliary-graphs-size} we have that the total number of vertices and the total number of edges in all auxiliary graphs $H$ of $G$ are $O(n)$ and $O(m)$ respectively. Then, again by Lemma \ref{lemma:auxiliary-graphs-size}, the total size (number of vertices and edges) of all auxiliary graphs $H_q^R$ for all $H$, computed in Step 3.4, is still $O(m)$ and they are also computed in $O(m)$ total time by Lemma \ref{lemma:auxiliary-graphs-construction}.
Now consider the $\mathit{split}$ operations. All these operations that occur during Step 3.3 for a specific auxiliary graph $G_r$ operate on the same tree $T$, which can be preprocessed once, as in Lemma \ref{lemma:split}, for all $\mathit{split}$ operations. Therefore, the total preprocessing time for all  $\mathit{split}$ operations is $O(n)$. Excluding the preprocessing time for $T$, a   $\mathit{split}(B,T)$ operation takes time proportional to the number of vertices in $B$. Therefore all  $\mathit{split}$ operations take $O(n)$ time in total by Lemmas \ref{lemma:blocks-size} and \ref{lemma:split}.
In Step 3.5.1 we examine the adjacency lists of the ordinary vertices $v \in H^R_q$ and find the corresponding blocks that contain at least such two ordinary vertices. Then we examine the adjacency lists of each such block. So, the adjacency lists of each vertex $v$ and each block that contains $v$ can be examined at most three times. Hence, Step 3.5.1 takes $O(n)$ time in total.
Finally, Steps 3.5.2 and 3.5.3 take $O(m)$ time in total by \cite{dfs:t} and Lemmas \ref{lemma:blocks-size} and \ref{lemma:refine}.
\end{proof}

\subsection{Queries}
\label{sec:queries}

Algorithm \textsf{FastVRB} computes the vertex-resilient blocks of the input digraph $G$ and stores them in the block forest $F$ of Section \ref{sec:blocks}, which makes it straightforward to test in constant time if two query vertices $v$ and $w$ are vertex-resilient. Here we show that if $v$ and $w$ are not vertex-resilient, then we can report a witness of this fact, that is, a strong articulation point $x$ such that $v$ and $w$ are not in the same strongly connected component of $G \setminus x$. Using this witness, it is straightforward to verify in $O(m)$ time that $v$ and $w$ are not vertex-resilient; it suffices to check that $v$ is not reachable from $w$ in $G \setminus x$ or vice versa.

To obtain this witness, we would like to apply Lemma \ref{lemma:VB-graph-path-2}, but this requires $v$ and $w$ to be in the same tree of the block forest. Fortunately, we can find the witness fast by applying Lemmas \ref{lemma:vertex-resilient-necessary} and \ref{lemma:blocks-paths}, which use information computed during the execution of \textsf{FastVRB}. We do that as follows.
First consider the simpler case where $v=s$. If Lemma \ref{lemma:vertex-resilient-necessary} does not hold for $s$ and $w$ in $D(s)$ then $d(w) \not= s$ is a strong articulation point that separates $s$ from $w$. Otherwise, $s=d(w)$, and $s$ and $w$ are both ordinary vertices in the auxiliary graph $H=G_s$. Then $s$ and $w$ cannot satisfy Lemma \ref{lemma:vertex-resilient-necessary} in $D^R_H(s)$, so $d_H^R(w)$ is a strong articulation point that separates $w$ from $s$. Now consider the case where $v,w \in V \setminus s$.
Suppose first that $v$ and $w$ do not satisfy Lemma \ref{lemma:vertex-resilient-necessary} in $D(s)$. Then $d(w)$ is not an ancestor of $v$ or $d(v)$  is not an ancestor of $w$ (or both). Assume, without loss of generality, that $d(w)$  is not an ancestor of $v$. By Lemma \ref{lemma:blocks-paths}, all paths from $v$ to $w$ pass through $d(w)$, so $d(w)$ is a strong articulation point that separates $v$ from $w$. On the other hand, if Lemma \ref{lemma:vertex-resilient-necessary} holds for $v$ and $w$ in $D(s)$, then $v$ and $w$ are both ordinary vertices in an auxiliary graph $H=G_r$, where $r=d(v)$ if $v=d(w)$, $r=d(w)$ if $w=d(v)$, and $r=d(v)=d(w)$ otherwise. By Lemma \ref{lemma:vertex-resilient-auxiliary}, $v$ and $w$ are not vertex-resilient in $H$. If they violate Lemma \ref{lemma:vertex-resilient-necessary} for $D_H^R(r)$ then we can find a strong articulation point that separates them as above. Finally, assume that Lemma \ref{lemma:vertex-resilient-necessary} holds for $v$ and $w$ in $D_H^R(r)$. Now $v$ and $w$ are both ordinary vertices in an auxiliary graph $H^R_q$. From the proof of Lemma \ref{lemma:FastVRB-correctness} we have that $q=d_H^R(v)$ or $q=d_H^R(w)$ and that $q$ is a strong articulation point that separates $v$ and $w$.

All the above tests can be performed in constant time. It suffices to store the dominator tree $D(s)$ of $G(s)$, and the dominator trees $D_H^R(r)$ of all auxiliary graphs $H^R=G_r^R$. The space required for these data structures is $O(n)$ by Lemma \ref{lemma:auxiliary-graphs-size}.

\ignore{
\begin{theorem}
\label{theorem:vrb}
Let $G$ be a digraph with $n$ vertices and $m$ edges. We can compute the vertex-resilient blocks of $G$ in $O(m+n)$ time. Given the vertex-resilient blocks of $G$, we can test in $O(1)$ time if any two vertices are vertex-resilient. Moreover, if the two vertices are not vertex-resilient, then we can report in $O(1)$ time a strong articulation point that separates them.
\end{theorem}
}

\begin{theorem}
\label{theorem:vrb}
Let $G$ be a digraph with $n$ vertices and $m$ edges. We can compute the vertex-resilient blocks of $G$ in $O(m+n)$ time and store them in a data structure of $O(n)$ space. Given this data structure, we can test in $O(1)$ time if any two vertices are vertex-resilient. Moreover, if the two vertices are not vertex-resilient, then we can report in $O(1)$ time a strong articulation point that separates them.
\end{theorem}

\section{Computing the $2$-vertex-connected blocks}
\label{section:2vc-blocks}

We can compute the $2$-vertex-connected blocks of the input digraph $G=(V,E)$ by applying Corollary \ref{cor:2vc-resilient} as follows.
Given the vertex-resilient blocks $\mathcal{B}$ and the $2$-edge-connected blocks $\mathcal{S}$ of $G$, we simply execute $\mathit{refine}(\mathcal{B}, \mathcal{S})$. This takes $O(n)$ time by Lemma \ref{lemma:refine}.
Also, since the $2$-vertex-connected blocks have a block forest representation, we can test if two given vertices are $2$-vertex-connected in $O(1)$ time as described in Section \ref{sec:blocks}.

If we only wish to answer queries of whether two vertices $v$ and $w$ are  $2$-vertex-connected, without computing explicitly the $2$-vertex and the $2$-edge-connected blocks, then we can use a simpler alternative, as suggested by Lemma \ref{lemma:2vc-resilient-sb}. First, we test if $v$ and $w$ are vertex-resilient in $O(1)$-time as in Section \ref{sec:queries}, and if they are not, then we can report a strong articulation point that separates them. If, on the other hand, $v$ and $w$ are vertex-resilient then we need to check if $G$ contains $(v,w)$ or $(w,v)$ as a strong bridge. We can do this easily using the same information as in Section \ref{sec:queries}, namely  the dominator tree $D(s)$ of $G(s)$, and the dominator trees $D_H^R(r)$ of all auxiliary graphs $H^R=G_r^R$. For instance, if $(v,w)$ is a strong bridge in $G$, then it will appear as an edge in one of the dominator trees. Therefore, it suffices to mark the edges of dominator trees that are strong bridges, and then check if $v$ is the parent of $w$ or $w$ is the parent of $v$ in $D(s)$ or in $D^R_H(r)$, where $H=G_r$ is the auxiliary graph of $G$ such that $r=d(v)$ if $v=d(w)$, $r=d(w)$ if $w=d(v)$, and $r=d(v)=d(w)$ otherwise.

\begin{theorem}
\label{theorem:2vcb}
Let $G$ be a digraph with $n$ vertices and $m$ edges. We can compute the $2$-vertex-connected blocks of $G$ in $O(m+n)$ time and store them in a data structure of $O(n)$ space. Given this data structure, we can test in $O(1)$ time if any two vertices are $2$-vertex-connected. Moreover, if the two vertices are not $2$-vertex-connected, then we can report in $O(1)$ time a strong articulation point or a strong bridge that separates them.
\end{theorem}

\ignore{
\begin{theorem}
\label{theorem:2vcb}
Let $G$ be a digraph with $n$ vertices and $m$ edges. We can compute the $2$-vertex-connected blocks of $G$ in $O(m+n)$ time. Also, given the  $2$-vertex-connected blocks of $G$, we can test in $O(1)$ time if any two vertices are $2$-vertex-connected.
\end{theorem}
}

\section{Sparse certificate for the vertex-resilient blocks and the $2$-vertex-connected blocks}
\label{section:sparse-certificate}

Here we show how to extend Algorithm \textsf{FastVRB} so that it also computes in linear time a sparse certificate for the vertex-resilient and the $2$-vertex-connected relations. That is, we compute a subgraph $C(G)$ of the input graph $G$ that has $O(n)$ edges
and maintains the same vertex-resilient and $2$-vertex-connected blocks as the input graph.
We can achieve this by applying the same approach we used in \cite{2ECB} for computing  a sparse certificate for the $2$-edge-connected blocks.

As in Section \ref{section:vertex-resilient-blocks} we can assume without loss of generality that $G$ is strongly connected, in which case subgraph $C(G)$ will also be strongly connected.
The certificate uses the concept of \emph{independent spanning trees} \cite{domcert}. A spanning tree $T$ of a flow graph $G(s)$ is a tree with root $s$ that contains a path from $s$ to $v$ for all vertices $v$.  Two spanning trees $B$ and $R$ rooted at $s$ are \emph{independent} if for all $v$, the paths from $s$ to $v$ in $B$ and $R$ share only the dominators of $v$. Every flow graph $G(s)$ has two such spanning trees, computable in linear time \cite{domcert}. Moreover, the computed spanning trees are \emph{maximally edge-disjoint}, meaning that the only edges they have in common are the bridges of $G(s)$.

During the execution of Algorithm \textsf{FastVRB}, we maintain a list (multiset) $L$ of the edges to be added in $C(G)$. The same edge may be inserted into $L$ multiple times, but the total number of insertions will be $O(n)$. Then we can use radix sort to remove duplicate edges in $O(n)$ time.
We initialize $L$ to be empty.
During Step 1 of Algorithm \textsf{FastVRB} we compute two independent spanning trees, $B(G(s))$ and $R(G(s))$ of $G(s)$ and insert their edges into $L$.
Next, in Step 3.1 we compute two independent spanning trees $B(H^R(r))$ and $R(H^R(r))$ for each auxiliary graph $H^R(r)$.
For each edge $(u,v)$ of these spanning trees, we insert a corresponding edge into $L$ as follows. If both $u$ and $v$ are ordinary vertices in $H^R(r)$, we insert $(u,v)$ into $L$ since it is an original edge of $G$.
Otherwise, $u$ or $v$ is an auxiliary vertex and we insert into $L$ a corresponding original edge of $G$. Such an original edge can be easily found during the construction of the auxiliary graphs.
Finally, in Step 3.5, we compute two spanning trees for every connected component $S_i$ of each auxiliary graph $H^R_q\setminus q$ as follows. Let $H_{S_i}$ be the subgraph of $H_q$ that is induced by the vertices in $S_i$. We choose an arbitrary vertex $v \in S_i$ and compute a spanning tree of $H_{S_i}(v)$ and a spanning tree of $H^R_{S_i}(v)$. We insert in $L$ the original edges that correspond to the edges of these spanning trees.

\begin{lemma}
The sparse certificate $C(G)$ has the same vertex-resilient blocks and $2$-vertex-connected blocks as the input digraph $G$.
\end{lemma}
\begin{proof}
We first argue that the execution of Algorithm \textsf{FastVRB} on $C(G)$ and produces the same vertex-resilient blocks as the execution of Algorithm \textsf{FastVRB} on $G$. The correctness of Algorithm \textsf{FastVRB} implies that it produces the same result regardless of the choice of start vertex $s$. So we assume that both executions choose the same start vertex $s$. We will refer to the execution of Algorithm \textsf{FastVRB} with input $G$ (resp. $C(G)$) as \textsf{FastVRB}$(G)$ (resp. \textsf{FastVRB}$(C(G))$).

First we note that $C(G)$ is strongly connected since it contains a spanning tree of $G(s)$ and a spanning tree for the reverse of each auxiliary graph $G_r$. Moreover, the fact that $C(G)$ contains two independent spanning trees of $G$ implies that $G$ and $C(G)$ have the same dominator tree with respect to the start vertex $s$ that are computed in Step 1. Hence, the auxiliary graphs computed in Step 2 of Algorithm \textsf{FastVRB} have the same sets of ordinary and auxiliary vertices in both executions \textsf{FastVRB}$(G)$ and \textsf{FastVRB}$(C(G))$. Hence, Step 3.1 computes the same dominator trees $D_H(r)$ and  $D^R_H(r)$ in both executions, and therefore Steps 3.2 and 3.3 compute the same blocks.
The same argument as in Steps 1 and 2 implies that both executions \textsf{FastVRB}$(G)$ and \textsf{FastVRB}$(C(G))$ compute in Step 3.4 auxiliary graphs $H_q^R$ with the same sets of ordinary and auxiliary vertices.
Finally, by construction, the strongly connected components of each auxiliary graph $H^R_q\setminus q$ are the same in both executions of \textsf{FastVRB}$(G)$ and \textsf{FastVRB}$(C(G))$.

We conclude that \textsf{FastVRB}$(G)$ and \textsf{FastVRB}$(C(G))$ compute the same vertex-resilient blocks as claimed. Next, observe that since the independent spanning trees computed in Steps 1 and 3.1 of the extended version of  \textsf{FastVRB} are maximally edge-disjoint, $C(G)$ maintains the same strong bridges as $G$. Then, by Corollary \ref{cor:2vc-resilient}, $C(G)$ also has the same $2$-vertex-connected blocks as $G$.
\end{proof}

\section{Concluding remarks }
\label{section:concluding}

We presented the first linear-time algorithms for computing the vertex-resilient and the $2$-vertex-connected relations among the vertices of a digraph.
We showed how to represent these relations with a data structure of $O(n)$ size, so that it is straightforward to check in constant time if any two vertices are vertex-resilient or $2$-vertex-connected.
Moreover, if the answer to such a query is negative, then we can provide a witness of this fact in constant time, i.e., a vertex (strong articulation point) or an edge (strong bridge) of $G$ that separates the two query vertices.
An experimental study of the algorithms described in this paper is presented in \cite{LuigiGILP15}, where it is shown that they perform very well in practice on very large graphs (with millions of vertices and edges).
We leave as an open question if the $2$-edge-connected or the $2$-vertex-connected components of a digraph can be computed faster than $O(n^2)$.


\end{document}